\documentclass[11p,reqno]{amsart}

\usepackage[T1]{fontenc}
\usepackage{graphicx}
\usepackage{amssymb,amsthm,amsmath,mathrsfs,bm,braket,marginnote}
\usepackage{enumerate}
\usepackage{appendix}
\usepackage[colorlinks=true, pdfstartview=FitV, linkcolor=blue, citecolor=blue, urlcolor=blue]{hyperref}
\usepackage{multirow}
\usepackage{cases}

\usepackage[symbol]{footmisc}

\usepackage{pgf}
\usepackage{pgfplots}
\usepackage{tikz}
\usetikzlibrary{arrows,calc}
\usepackage{verbatim}
\usetikzlibrary{decorations.pathreplacing,decorations.pathmorphing}
\usepackage[numbers,sort&compress]{natbib}
\usepackage{dsfont}

\numberwithin{equation}{section}
\newtheorem{theorem}{Theorem}[section]
\newtheorem{lemma}[theorem]{Lemma}
\newtheorem{definition}[theorem]{Definition}
\newtheorem{remark}[theorem]{Remark}

\newtheorem{proposition}[theorem]{Proposition}

 \reversemarginpar
 
 \newcommand{\mt}{\mathbf{t}}
 
 \newcommand{\p}{\partial}

\newcommand{\teta}{\tilde{\eta}}
\newcommand{\tth}{\tilde{H}}

\begin{document}

\title[matrix-valued $\theta$-deformed bi-orthogonal polynomials and Toda]{Matrix-valued $\theta$-deformed bi-orthogonal polynomials, Non-commutative Toda theory and B\"acklund transformation}

\date{}

\dedicatory{}


\author{Claire Gilson}
\address{School of Mathematics and Statistics, University of Glasgow, Glasgow G12 8SQ, UK}
\email{claire.gilson@glasgow.ac.uk}

\author{Shi-Hao Li}
\address{Department of Mathematics, Sichuan University, Chengdu, 610064, PR China}
\email{shihao.li@scu.edu.cn}

\author{Ying Shi}
\address{School of Science, Zhejiang University of Science and Technology, Hangzhou, 310023, PR China}
\email{yingshi@zust.edu.cn}

\begin{abstract}
\noindent 
This paper is devoted to revealing the relationship between matrix-valued $\theta$-deformed bi-orthogonal polynomials and non-commutative Toda-type hierarchies. In this procedure, Wronski quasi-determinants are widely used and play the role of non-commutative $\tau$-functions. At the same time, B\"acklund transformations are realized by using a moment modification method and non-commutative $\theta$-deformed Volterra hierarchies are obtained, which contain the known examples of the Itoh-Narita-Bogoyavlensky lattices and the fractional Volterra hierarchy. 
\end{abstract}

\maketitle

\section{Introduction}

The studies of connections between orthogonal polynomials and integrable systems doesn't only promote the development in their own respective directions, but stimulates disciplinary researches on random matrices, combinatorics, probability and so on. A famous example in this field is the connection between standard orthogonal polynomials and the Toda equation \cite{deift00,forrester10}. Starting from a non-negative weight function $\omega(x)$, one can define an inner product 
\begin{align*}
\langle \cdot,\cdot\rangle: \mathbb{R}[x]\times\mathbb{R}[x]\to\mathbb{R}, \quad \langle f(x),g(x)\rangle=\int_{\mathbb{R}} f(x)g(x)\omega(x)dx,
\end{align*}
and a sequence of monic orthogonal polynomials $\{P_n(x)\}_{n\in\mathbb{N}}$ is then defined by the orthogonality $$\langle P_n(x),P_m(x)\rangle=h_n\delta_{n,m}$$ for some non-singular normalization constant $h_n$, where deg $P_n=n$. It is then known that there exists a three-term recurrence relation for the orthogonal polynomials 
\begin{align}\label{ttrr}
xP_n(x)=P_{n+1}(x)+a_nP_n(x)+b_nP_{n-1}(x),\quad P_{-1}(x)=0,\,P_0(x)=1
\end{align}	
for some coefficients $a_n$ and $b_n$.
This recurrence relation plays a role as a spectral problem and is a key ingredient to connect with the Toda equation. If we assume that there exist time evolutions in the weight function such that 
$$\omega(x)\mapsto\omega(x;\mt):=\exp\left(\sum_{i=1}^\infty t_ix^i\right)\omega(x),$$ 
then the orthogonal polynomials are dependent with time flows, and the Toda equation could be derived by the compatibility condition of the spectral equation and evolution part. Therefore, given a sequence of orthogonal polynomials, we could construct integrable structures for classical integrable systems, such as a Lax pair, wave functions, dressing structures, $\tau$-functions, and symmetries \cite{adler95}. These exactly integrable structures lay a solid foundation for applications into different branches in mathematics like integrable geometry \cite{mansfield13}, integrable combinatorics \cite{difrancesco14}, integrable probability \cite{okounkov01} etc. This idea was later generalized to find more connections among bi-orthogonal polynomials, integrable systems, random matrices and related fields, see e.g.  \cite{adler97,adler99,chang18,li19} and references therein.

In the late 1940s, Krein carried out research into matrix-valued measures and matrix-valued orthogonal polynomials \cite{krein49}. 
In recent years, as matrix-valued orthogonal polynomials play an important role in solving problems such as quasi-birth-and-death processes \cite{grunbaum08} and two periodic Aztec diamond problems \cite{duits18}, they have gradually become an essential algebraic tool. 
In particular, as an important application of the quasi-determinant, it was shown in \cite{gelfand05} that matrix-valued orthogonal polynomials could be expressed in closed form, making different applications such as non-commutative Hermite-Pad\'e approximation \cite{doliwa22}, Wynn recurrence \cite{doliwa222} and non-commutative integrable systems possible. 
Regarding with connections between matrix-valued orthogonal polynomials and non-commutative integrable systems, studies are mainly on non-commutative Painlev\'e equation \cite{cafasso14} and the non-commutative Toda equation \cite{branquinho20,ismail19,li20}. 
Moreover, an attempt to understand other non-commutative Toda-type equations has been made by considering the matrix-valued Cauchy bi-orthogonal polynomials,   and the corresponding non-commutative version of the C-Toda lattice was recently proposed in \cite{li22}. 

This article will continue this thread by considering matrix versions of bi-orthogonal polynomials, thus establishing more connections between matrix-valued polynomials and non-commutative integrable systems. We start with a special bilinear form (in the scalar case)
\begin{align}\label{bf}
\langle f(x),g(x)\rangle_\theta:=\int_{\mathbb{R}}f(x)g(x^\theta)\omega(x)dx.
\end{align}
When $\theta\in\mathbb{Z}_+$, this bilinear form can be dated back to earlier works of Konhauser \cite{konhauser67} and Carlitz \cite{carlitz68} as a generalization of the Laguerre polynomials. In fact,  \eqref{bf} is valid for arbitrary $\theta\in\mathbb{R}_+$, and the corresponding bi-orthogonal polynomials were used to describe a Hermitian matrix model with additional interaction, which is called the Muttalib-Borodin model \cite{borodin98,muttalib95}. 
One can see that $\theta$ is a free parameter in \eqref{bf} compared with the standard inner product, and thus the corresponding bi-orthogonal polynomials are sometimes considered as a $\theta$-deformation of the standard orthogonal polynomials. On the other hand, as standard orthogonal polynomials are connected with the Toda equation, it is expected that these $\theta$-deformed bi-orthogonal polynomials should correspond to a $\theta$-deformed integrable hierarchy. This result is known from earlier private communication with Ipsen \cite{ipsen}. 

In Sec \ref{sec2}, we demonstrate the corresponding non-commutative $\theta$-deformed Toda theory by rewriting the weight function into a matrix-valued weight function in \eqref{bf}. An important assumption on weight function, which we call the moment condition (see Definition \ref{def1.1}), is made to ensure the existence and uniqueness of the matrix-valued bi-orthogonal polynomials. Moreover, we show that the non-commutative $\theta$-deformed Toda hierarchy is the same as the non-commutative hungry Toda hierarchy (aka non-commutative Blaszak-Marciniak hierarchy, Kuperschmidt hierarchy etc) by making the use of matrix-valued bi-orthogonal polynomials when $\theta\in\mathbb{Z}_+$. In the commutative case, such a hierarchy could be viewed as a special situation of Kostant-Toda hierarchy, whose solutions were constructed by the Wronskian determinant \cite{kodama15}. 
Therefore,  in Section \ref{sec2.3}, we develop a  Wronski quasi-determinant technique and verify the solutions of the non-commutative Blaszak-Marciniak three-field equations. It is then found that for the general non-commutative Blaszak-Marcinak hierarchy, the non-commutative nonlinear variables can be written in terms of a single Wronski quasi-determinant with higher-order derivatives. Therefore, such a method could be regarded as a modification of Hirota's direct method into non-commutative integrable systems. Although the equations written using a single quasi-determinant are no longer bilinear, the ideas of Hirota's bilinear method could be extensively enlarged into non-commutative circumstances. Therefore, based on reduction techniques in Hirota's bilinear method, we apply the idea of B\"acklund transformation from the commutative case into the non-commutative case.  These details are discussed later in Sec. \ref{sec4}.

In Sec \ref{sec3}, we generalize our choice of $\theta$ to positive rational numbers. For $\theta=a/b$, where $a,b\in\mathbb{Z}_+$, we show that the corresponding matrix-valued bi-orthogonal polynomials satisfy an $(a+b+1)$-term recurrence relation. Thus it leads us to two-parameter deformed non-commutative equations. In the commutative case, the corresponding two-parameter deformed theory is related to the extended bigraded Toda hierarchy studied in \cite{carlet04,carlet06,dubrovin04} arising from the theory of Frobenius manifolds and geometric structures. 
In \cite{li11}, solutions of the bigraded Toda hierarchy were given by using string orthogonal polynomials, which are wave functions for the 2-dimensional Toda hierarchy \cite{adler97}. Therefore, works in \cite{li11} is the realization of a bigraded Toda hierarchy as a reduction of 2d-Toda theory. In this paper, we provide an explanation for such a bigraded Toda hierarchy by using matrix-valued $\theta$-deformed bi-orthogonal polynomials and characterize its solution in closed form by making use of block Wronski quasi-determinants. We show that $\theta$-deformed bi-orthogonal polynomials act exactly as the wave function for an extended bigraded Toda hierarchy.

In Sec \ref{sec4},  we mainly use the moment reduction technique to grade the matrix-valued orthogonal polynomial space, and find out the B\"acklund transformation for the $\theta$-deformed integrable hierarchies. 
In Sec. \ref{sec2}, we showed that a Wronski quasi-determinant could be used to construct solutions for the Blaszak-Marciniak equation. According to a property of Wronski (quasi-)determinants, it is known that (quasi-)determinants of the same order but with different phases can be regarded as solutions to the same equation. Moreover, these solutions could be linked to a simple equation according to Hirota's idea of bilinear B\"acklund transformation \cite[\S 4]{hirota04}. Such an idea has been widely applied to the correspondence between the Toda and the Lotka-Volterra equation (i.e. Kac-van Moerbeke lattice) \cite{gesztesy93}, and later was used in orthogonal polynomials by making constraints on the weight function \cite{tsujimoto00}. 

Revisiting the connection between orthogonal polynomials and the Toda equation mentioned at the beginning,  we can set the weight function symmetrically so that moments admit the form 
\begin{align*}
m_{i,j}=\langle x^i,x^j\rangle=\left\{\begin{array}{ll}
m_{i+j},&\text{if $i+j$ is even},\\
0,&\text{if $i+j$ is odd}.
\end{array}
\right.
\end{align*}
Then corresponding orthogonal polynomials are symmetric, and their normalization factors are dependent with $\tau_n^{(0)}=\det(d_{i+j})_{i,j=0}^{n-1}$ and $\tau_n^{(1)}=\det(d_{i+j+1})_{i,j=0}^{n-1}$  where $d_i=m_{2i}$. The Wronski technique tells us that if we assume that $\p_t d_i=d_{i+1}$, then each $\tau_n^{(0)}$ and $\tau_n^{(1)}$ are solutions of the Toda equation. Therefore, the equation derived by using symmetric orthogonal polynomials will be an integrable equation connecting different solutions of the Toda equation. In Sec \ref{sec4.1}, we deduce in detail the B\"acklund transformation corresponding to the non-commutative Blaszak-Marcininak lattice when $\theta$ is a positive integer. The original solution space of $\tau$-functions is divided into $\theta+1$ different families, together with matrix-valued polynomial space. It is shown that the polynomial space could be graded as a direct sum of equivalence classes, in which the powers of polynomials are the same modulu $\theta+1$. As a result, the B\"acklund transformation of the non-commutative Blaszak-Marciniak lattice can be written as addition and multiplication forms of the Itoh-Narita-Bogoyavlensky (INB) lattice in the non-commutative version. Since the INB lattice hierarchy can be regarded as a discretization of the Gelfand-Dickey hierarchy, we understand that the $\theta$-deformed integrable hierarchy can be regarded as the discrete Gelfand-Dickey flows under the perspective of orthogonal polynomials. In Sec. \ref{sec4.2}, we again use the Wronski quasi-determinant technique to verify solutions of a specific INB lattice, and demonstrate that the graded $\tau$-functions are simply connected by non-commutative Jacobi identities. Moreover, the case where $\theta\in\mathbb{Q}_+$ is discussed in Sec. \ref{sec4.3}. We find that the spectral problem in this case corresponds to the fractional Volterra hierarchy proposed in \cite{liu18}, and under certain time flows, we obtain the corresponding integrable equations.

The highlights of this article are the following:
\begin{enumerate}
\item Matrix-valued $\theta$-deformed bi-orthogonal polynomials are proposed, and corresponding non-commutative integrable systems are obtained, with Lax pairs and solutions;
\item The direct method of Wronski quasi-determinants is developed. We verify solutions of several non-commutative integrable systems by using the quasi-Wronski technique;
\item The B\"acklund transformation can be understood as a gradation of the solution space. We make the moment reduction approach to grade wave function space and solution space, and thus realize the corresponding B\"acklund transformation.
\end{enumerate}

\section{Matrix-valued bi-orthogonal polynomials and Recurrence relation}\label{sec2}	

Before we work on the matrix-valued bi-orthogonal polynomials, we need to introduce a matrix-valued Radon measure $\mu$: $(-\infty,\infty)\to\mathbb{R}^{p\times p}$. Firstly, according to   
Riesz-Markov-Kakutani representation theory, it is known that for any  positive linear functional $\psi$ on $C_c(\mathbb{R})$ (the space of continuous compactly supported real-valued functions on $\mathbb{R}$), there is a unique Radon measure $\mu$ on $\mathbb{R}$ such that
\begin{align}\label{int}
\psi(f)=\int_{\mathbb{R}} f(x)d\mu(x).
\end{align}
Therefore, for any matrix-valued polynomials $f(x)\in\mathbb{R}^{p\times p}[x]$, the integration \eqref{int} is well-defined for a Radon measure $\mu$. Moreover, if we normalize the Radon measure $\mu(\mathbb{R})=\mathbb{I}_p$, where $\mathbb{I}_p$ is a $p\times p$ identity matrix, then according to the Radon-Nikodym theorem, the normalized Radon measure $\mu$ is related to a matrix-valued weight function $W(x)$ such that $d\mu(x)=W(x)dx$. Please refer to \cite{damanik07} for details.

Therefore, the weight function $W(x)$ can induce a bilinear form
\begin{align}\label{innerproduct}
{\langle \cdot,\cdot\rangle}_\theta: \mathbb{R}^{p\times p}[x]\times \mathbb{R}^{p\times p}[x]\rightarrow \mathbb{R}^{p\times p},\quad 
{\langle f(x),g(x)\rangle}_\theta=\int_{\mathbb{R}}f(x)W(x){g^\top(x^\theta)} dx,
\end{align}
with dependence of $\theta\in\mathbb{Z}_+$.  Here $\top$ represents the transpose of a matrix. To make the bilinear form well-defined, we need to make some assumptions about $W(x)$. Similar to the weight function discussed in \cite{duits18}, the matrix-valued weight function is not necessarily symmetric or Hermitian. However, to ensure the existence and uniqueness of corresponding matrix-valued bi-orthogonal polynomials, details of requirements on $W(x)$ are addressed in Definition \ref{def1.1}. Firstly, we state some properties of the bilinear form.

\begin{proposition}
The bilinear form \eqref{innerproduct}  has the following properties:
\begin{enumerate}
\item {\textbf{Bimodule structures}}. For any $L_1,\,L_2,\,R_1,\,R_2\in\mathbb{R}^{p\times p}$ and $f_1(x),f_2(x),g_1(x),g_2(x)\in\mathbb{R}^{p\times p}[x]$, we have
 \begin{align}\label{bimodule}
 \begin{aligned}
& {\langle L_1f_1(x)+L_2f_2(x),g(x)\rangle}_\theta=L_1{\langle f_1(x), g(x)\rangle}_\theta+L_2{\langle f_2(x),g(x)\rangle}_\theta,\\
& {\langle f(x),R_1 g_1(x)+R_2 g_2(x)\rangle}_\theta={\langle f(x),g_1(x)\rangle}_\theta R_1^\top+{\langle f(x),g_2(x)\rangle}_\theta R_2^\top.
\end{aligned}
\end{align}

\item {\textbf{Quasi-symmetry property}.} For any $f(x),g(x)\in\mathbb{R}^{p\times p}[x]$, it holds that
\begin{align}\label{qs}
&\langle x^\theta f(x), g(x)\rangle_\theta = \langle f(x),x g(x)\rangle_\theta,\quad \text{~for~} \theta\in\mathbb{Z}_+.
\end{align}
\end{enumerate}
\end{proposition}

Moreover, the bilinear form \eqref{innerproduct} can induce a family of monic bi-orthogonal polynomial  sequences $\{P_n(x), Q_n(x)\}_{n\in\mathbb{N}}$ such that
\begin{align}\label{or}
\langle P_n(x),Q_m(x)\rangle_\theta=H_n\delta_{n,m}, 
\end{align} 
where $H_n\in\mathbb{R}^{p\times p}$ is a nonsingular normalization factor and deg $P_n$=deg $Q_n$=$n$. Since $\mathbb{R}^{p\times p}[x]$ is a free module and $\{x^k\mathbb{I}_p\}_{k\in\mathbb{N}}$ form its basis, we know that $P_n(x)$ can be expanded as 
\begin{align}\label{p-poly}
P_n(x)=\mathbb{I}_px^n+\xi_{n,n-1}x^{n-1}+\cdots+\xi_{n,0},\quad \xi_{n,j}\in\mathbb{R}^{p\times p}, \quad j=0,\cdots,n-1.
\end{align}
Moreover, according to bimodule property \eqref{bimodule}, the orthogonal condition \eqref{or} is equivalent to 
\begin{align}\label{or-2}
\begin{aligned}
{\langle P_n(x),x^j\mathbb{I}_p\rangle}_\theta=0,\quad 0\leq j\leq n-1.
\end{aligned}
\end{align}
Therefore, if we denote moments
\begin{align}\label{moments}
m_{i+j\theta}=\langle x^i\mathbb{I}_p,x^j\mathbb{I}_p\rangle_\theta=\int_{\mathbb{R}}(x^i\mathbb{I}_p)W(x)(x^j\mathbb{I}_p)^\theta dx=\int_{\mathbb{R}}x^{i+j\theta}W(x)dx,
\end{align}
then the orthogonal condition \eqref{or-2} is equivalent to a linear system with matrix-valued coefficients
\begin{align}\label{ls}
\xi_{n,0}m_{j\theta}+\xi_{n,1}m_{1+j\theta}+\cdots+\xi_{n,n-1}m_{n-1+j\theta}=-m_{n+j\theta},\quad j=0,\cdots,n-1.
\end{align}
It is noted that the existence and uniqueness of bi-orthogonal polynomials $\{P_n(x)\}_{n\in\mathbb{N}}$ are equivalent to the existence and uniqueness of solutions in \eqref{ls}.
Therefore, we make the following assumptions about the weight function $W(x)$, which we call the moment condition.
\begin{definition}\label{def1.1}
The weight function $W(x)$ satisfies the moment condition if
\begin{enumerate}
\item all moments $m_{i+j\theta}$ exist and are finite;
\item the moment matrices $\left(
m_{i+j\theta}
\right)_{i,j=0,1,\cdots}$ are invertible.
\end{enumerate}
\end{definition}
It is known from Proposition \ref{p-ls} that if $W(x)$ satisfies the moment condition, then coefficients of $P_n(x)$ are given by\footnote[3]{~For self-consistent, we give the definitions and basic properties of quasi-determinants in the appendix.}
\begin{align}\label{xi}
\xi_{n,j}=-\left(
m_{n},m_{n+\theta},\cdots,m_{n+(n-1)\theta}
\right)\left(\begin{array}{cccc}
m_{0}&m_{\theta}&\cdots&m_{(n-1)\theta}\\
m_1&m_{\theta+1}&\cdots&m_{(n-1)\theta+1}\\
\vdots&\vdots& &\vdots\\
m_{n-1}&m_{n-1+\theta}&\cdots&m_{n-1+(n-1)\theta}
\end{array}
\right)^{-1}e_{j+1}^\top,
\end{align}
where 
\begin{align}\label{ue}
e_j=(0,\cdots,\mathbb{I}_p,\cdots,0)
\end{align} is the block unit vector, whose $j$-th element is the unity $\mathbb{I}_p$ and the others are zeros. Therefore, by substituting \eqref{xi} into the expansion \eqref{p-poly}, we have the quasi-determinant formula
 \begin{align}\label{pn}
P_n(x)=\left|\begin{array}{cccc}
m_{0}&\cdots&m_{(n-1)\theta}&\mathbb{I}_p\\
\vdots& &\vdots&\vdots\\
m_{n-1}&\cdots&m_{n-1+(n-1)\theta}&x^{n-1}\mathbb{I}_p\\
m_{n}&\cdots&m_{n+(n-1)\theta}&\boxed{x^n\mathbb{I}_p}
\end{array}
\right|.
\end{align}

On the other hand, if we assume that  
\begin{align}\label{q-poly}
Q^\top_n(x)=\mathbb{I}_px^n+\eta_{n,n-1} x^{n-1}+\cdots+\eta_{n,0},\quad \eta_{n,j}\in\mathbb{R}^{p\times p},\quad j=0,\cdots,n-1,
\end{align}
then the orthogonal relation \eqref{or} gives the linear system
\begin{align}\label{ls-eta}
m_{j}{\eta_{j, 0}}+m_{j+\theta}{\eta_{j, 1}}+\cdots+m_{j+(n-1)\theta}{\eta_{j,n-1}}=-m_{j+n\theta},\quad j=0,\cdots,n-1.
\end{align}
Therefore, coefficients of $Q^\top_n(x)$ are given by 
\begin{align}\label{eta}
\eta_{n,j}=-e_{j+1}
\left(\begin{array}{ccccc}
m_{0}&m_{\theta}&\cdots&m_{(n-1)\theta}\\
m_1&m_{1+\theta}&\cdots&m_{1+(n-1)\theta}\\
\vdots&\vdots& &\vdots\\
m_{n-1}&m_{n-1+\theta}&\cdots&m_{n-1+(n-1)\theta}
\end{array}
\right)^{-1}
\left(\begin{array}{cc}
m_{n\theta}\\
m_{1+n\theta}\\
\vdots\\
m_{n-1+n\theta}
\end{array}
\right),
\end{align}
and $Q_n^\top(x)$ admits the quasi-determinant formula
\begin{align}\label{qn}
{Q^\top_n(x)}=\left|\begin{array}{cccc}
m_{0}&\cdots&m_{(n-1)\theta}&m_{n\theta}\\
m_1&\cdots&m_{(n-1)\theta+1}&m_{n\theta+1}\\
\vdots& &\vdots&\vdots\\
m_{n-1}&\cdots&m_{n-1+(n-1)\theta}&m_{n-1+n\theta}\\
\mathbb{I}_p&\cdots&x^{n-1}\mathbb{I}_p&\boxed{x^n\mathbb{I}_p}
\end{array}
\right|.\end{align}

To conclude, we have the following definition.
\begin{definition}
With bilinear form \eqref{innerproduct} in which $W(x)$ satisfies the moment condition, a family of matrix-valued bi-orthogonal polynomials $\{P_n(x), Q_n(x)\}_{n\in\mathbb{N}}$ are defined by 
\begin{align*}
\langle P_n(x),Q_m(x)\rangle_\theta=H_n\delta_{n,m},
\end{align*}
where $P_n(x)$ and $Q(x)$ are given by \eqref{pn} and \eqref{qn} respectively, and
\begin{align}\label{nf}
H_n=\left|\begin{array}{cccc}
m_{0}&\cdots&m_{(n-1)\theta}&m_{n\theta}\\
\vdots&&\vdots&\vdots\\
m_{n-1}&\cdots&m_{n-1+(n-1)\theta}&m_{n-1+n\theta}\\
m_{n}&\cdots&m_{n+(n-1)\theta}&\boxed{m_{n+n\theta}}
\end{array}
\right|.
\end{align}
\end{definition}

\subsection{Recurrence relations for $\theta\in\mathbb{Z}_+$}
Hereafter, we are going to discuss the recurrence relations for $\{P_n(x), Q_n(x)\}_{n\in\mathbb{N}}$ when $\theta\in\mathbb{Z}_+$.

\begin{proposition}
The matrix-valued bi-orthogonal polynomials $\{P_n(x), Q_n(x)\}_{n\in\mathbb{N}}$ satisfy the following recurrence relations
\begin{subequations}
\begin{align}
x^\theta P_n(x)=P_{n+\theta}(x)+\sum^{n+\theta-1}_{j=n-1}\alpha_{n, j}P_j(x),\label{re-1}\\
x Q_n(x)=Q_{n+1}(x)+\sum^{n}_{j=n-\theta}\beta_{n, j}Q_j(x),\label{re-2}
\end{align}
\end{subequations}
for $\theta\in\mathbb{Z}_+$ and some certain $\alpha_{n, j},\,\beta_{n, j} \in\mathbb{R}^{p\times p}$.
\end{proposition}

\begin{proof}
Since $\{P_n(x)\}_{n\in\mathbb{N}}$ form a  basis of the left module $\mathbb{R}^{p\times p}[x]$ under the bilinear form \eqref{innerproduct}, it is known that any polynomial in $\mathbb{R}^{p\times p}[x]$ can be written as a left linear combination of $\{P_n(x)\}_{n\in\mathbb{N}}$. Therefore, we have
\begin{align*}
x^\theta P_{n}(x)=P_{n+\theta}(x)+\sum_{i=0}^{n+\theta-1}\alpha_{n,i}P_i(x),\quad \alpha_{n,i}\in\mathbb{R}^{p\times p},
\end{align*}
and the recurrence coefficients
\begin{align}\label{rc}
\alpha_{n,i}=\langle x^\theta P_n(x),Q_i(x)\rangle_\theta\cdot H_i^{-1}=\langle P_n(x),xQ_i(x)\rangle_\theta \cdot H_i^{-1},\quad i=0,\cdots,n+\theta-1,
\end{align}
where the last equality is due to the quasi-symmetry property \eqref{qs}.
When $i<n-1$, we know that deg $xQ_i<n$, and according to \eqref{or}, the bilinear form is definitely zero. It means that $\alpha_{n,i}=0$ when $0\leq i< n-1$, and the recurrence relation \eqref{re-1} is obtained. 
The recurrence relation \eqref{re-2} for $\{Q_{n}(x)\}_{n\in\mathbb{N}}$ can be proved similarly.
\end{proof}
\begin{remark}
The recurrence relations \eqref{re-1} and \eqref{re-2} can alternatively be written in a matrix form.
If we denote
\begin{align}\label{wave1}
\Phi=\left(\begin{array}{c}
P_0(x)\\
P_1(x)\\
\vdots
\end{array}
\right),\quad \Psi=\left(\begin{array}{c}
Q_0(x)\\
Q_1(x)\\
\vdots
\end{array}
\right),
\end{align}
then \eqref{re-1} can be written as 
\begin{align*}
x^\theta \Phi=L\Phi,\quad L=\Lambda^{\theta}+a_{\theta-1}\Lambda^{\theta-1}+\cdots+a_{-1}\Lambda^{-1},
\end{align*}
where $\{a_i\}_{i=-1,\cdots,\theta-1}$ are block diagonal matrices $a_i=\text{diag}(\alpha_{0,i},\alpha_{1,i+1},\alpha_{2,i+2},\cdots)$, $\Lambda$ is a block shift operator
\begin{align*}
\Lambda=\left(\begin{array}{ccccc}
0&\mathbb{I}_p&0&0&\cdots\\
0&0&\mathbb{I}_p&0&\cdots\\
0&0&0&\mathbb{I}_p&\cdots\\
\vdots&\vdots&\vdots&\vdots&\ddots
\end{array}
\right),
\end{align*}
and $\Lambda^{-1}$ is defined by the transpose of $\Lambda$\footnote[4]{~It should be noted that $\Lambda^{-1}$ is not the inverse of $\Lambda$. We use the notation $\Lambda^{-1}$ to denote the transpose of $\Lambda$. In our later computations, we use the positive power and negative power of $\Lambda$ to indicate the block upper triangular part and block lower triangular part of a matrix, respectively. }.
Similarly, the recurrence \eqref{re-2} could be written by
\begin{align*}
x\Psi=M\Psi,\quad M=\Lambda+\beta_{0}\Lambda^0+\cdots+\beta_{\theta}\Lambda^{-\theta},
\end{align*}
where $\beta_i=\text{diag}(\beta_{i,0},\beta_{i+1,1},\beta_{i+2,2},\cdots)$ for $i=0,\cdots,\theta$.
\end{remark}
The following proposition states that those recurrence coefficients $\alpha_{n,j}$ and $\beta_{n,j}$ could be written in terms of quasi-determinants.
\begin{proposition}
Recurrence coefficients $\alpha_{n,j}$  could be written in terms of quasi-determinants and
\begin{align*}
\alpha_{n,j}=\left(Z_{n,j}+\sum_{k=n-1}^{j-1}Z_{n,k}\eta_{j,k}\right)H_j^{-1},\quad j=n-1,\cdots,n+\theta-1,
\end{align*}
where 
\begin{align}\label{znj}
Z_{n,j}={\langle P_n(x), x^{j+1}\mathbb{I}_p\rangle}_\theta
=\left|
\begin{array}{ccccc}
m_{0}&m_{\theta}&\cdots&m_{(n-1)\theta}&m_{(j+1)\theta}\\
\vdots&\vdots& &\vdots&\vdots\\
m_{n-1}&m_{n-1+\theta}&\cdots&m_{n-1+(n-1)\theta}&m_{n-1+(j+1)\theta}\\
m_{n}&m_{n+\theta}&\cdots&m_{n+(n-1)\theta}&\boxed{m_{n+(j+1)\theta}}
\end{array}
\right|,
\end{align}
and $\eta_{j,k}$ is the coefficient of $Q_j(x)$ given in \eqref{eta}. Moreover, if we introduce the notation
\begin{align*}
Y_{n, j+\theta-1}= \langle {x^{j+\theta}\mathbb{I}_p,Q_n(x)\rangle}_\theta
=\left|
\begin{array}{ccccc}
m_{0}&m_{\theta}&\cdots&m_{(n-1)\theta}&m_{n\theta}\\
\vdots&\vdots& &\vdots&\vdots\\
m_{n-1}&m_{n-1+\theta}&\cdots&m_{n-1+(n-1)\theta}&m_{n-1+n\theta}\\
m_{j+\theta}&m_{j+2\theta}&\cdots&m_{j+n\theta}&\boxed{m_{j+(n+1)\theta}}
\end{array}
\right|,
\end{align*}
then $\beta^\top_{n,j}$ could be expressed by
\begin{align*}
\beta^\top_{n,j}=H_j^{-1}\cdot\left(
Y_{n,j+\theta-1}+\sum_{k=n-\theta}^{j-1}\xi_{j,k}Y_{n,k+\theta-1}
\right),\quad j=n-\theta,\cdots,n,
\end{align*}
where $\xi_{j,k}$ is the coefficient of $P_j(x)$ given in \eqref{xi}.
\end{proposition}
\begin{proof}
Here we only prove the quasi-determinant expression for $\alpha_{n,j}$,  that for $\beta_{n,j}^\top$ can be similarly verified. By substituting the expansion of $Q_j(x)$ in \eqref{q-poly} into the formula \eqref{rc}, we obtain  
\begin{align*}
\alpha_{n,j}=\left(\langle P_n(x),x^{j+1}\mathbb{I}_p\rangle_\theta+\sum_{k=0}^{j-1} \langle P_n(x),x^{k+1}\mathbb{I}_p\rangle_\theta\cdot \eta_{j,k}\right)H_j^{-1}.
\end{align*}
To simplify above equation, the notation $Z_{n,j}$ is thus introduced as the bilinear form $\langle P_n(x), x^{j+1}\mathbb{I}_p\rangle_\theta$, which has the quasi-determinant expression as in \eqref{znj}. Moreover, when $j=0,\cdots,n-2$, there are two identical columns in $Z_{n,j}$. According to Proposition \ref{p-equi}, we know that $Z_{n,j}=0$ for $j=0,\cdots,n-2$, indicating that recurrence coefficients are truncated.
\end{proof}
\begin{remark}
It should be remarked that $\{\alpha_{n,j}\}_{j=n-1}^{n+\theta-1}$ admit another expression. Since we can expand $P_n(x)$ in terms of \eqref{p-poly}, then by comparing coefficients in \eqref{re-1} on both sides, we have
\begin{align*}
&\alpha_{n,n+\theta-1}=\xi_{n,n-1}-\xi_{n+\theta,n+\theta-1},\\
&\alpha_{n,n+\theta-2}=\xi_{n,n-2}-\xi_{n+\theta,n+\theta-2}-\alpha_{n,n+\theta-1}\xi_{n+\theta-1,n+\theta-2},\\
&\cdots\cdots
\end{align*}
This formula is useful in the verification of solutions of the Blaszak-Marcininak three-field equations.
\end{remark}

\subsection{Time evolutions and non-commutative lattices}
In this part, we are going to discuss how to introduce time flows into the matrix-valued bi-orthogonal polynomials. 
We assume that there is a family of time variables $\mt=(t_1,t_2,\cdots)$ added into the weight function such that
\begin{align}\label{td}
W(x;\mt)=\exp\left(\sum_{i=1}^\infty t_{i}x^{i}\right)W(x).
\end{align}
Such deformation plays the role of a ladder operator in orthogonal polynomials theory since the action of $\p_{t_i}$ is equivalent to the action of $x^i$, i.e. $\p_{t_i}W(x;\mt)=x^iW(x;\mt)$. Under such evolution assumptions, we have 
\begin{align}\label{momentevo}
\p_{t_i}m_k=m_{k+i}
\end{align}
for the moments $\{m_k\}_{k\in\mathbb{N}}$ defined in \eqref{moments}.

Therefore, we have a time-deformed bilinear form 
\begin{align*}
{\langle f(x),g(x)\rangle}_\theta=\int_{\mathbb{R}}f(x)W(x;\mt){g^\top(x^\theta)} dx,
\end{align*}
and time-dependent matrix-valued bi-orthogonal polynomials can be defined by 
\begin{align}\label{t-or}
\langle P_n(x;\mt),Q_m(x;\mt)\rangle_\theta=H_n(\mt)\delta_{n,m}.
\end{align} 
The derivative formulas for the time-dependent matrix-valued bi-orthogonal polynomials are then investigated.
\begin{proposition}\label{proptp}
$\{P_n(x;\mt)\}_{n\in\mathbb{N}}$ satisfying \eqref{t-or}  have the evolution equation
\begin{align}\label{t-p}
\p_{t_\theta}P_n(x;\mt)=-\alpha_{n,n-1}P_{n-1}(x;\mt),
\end{align}
where $\alpha_{n,n-1}$ is the recurrence coefficient in \eqref{re-1}; and  $\{Q_n(x;\mt)\}_{n\in\mathbb{N}}$  satisfy the evolution equation
\begin{align}\label{t-q}
\p_{t_\theta}Q_n(x;\mt)=-\!\!\sum_{j=n-\theta}^{n-1}\beta_{n,j} Q_{j}(x;\mt),
\end{align}
where $\{\beta_{n,j}\}_{j=n-\theta}^{n-1}$ are recurrence coefficients in \eqref{re-2}.
\end{proposition}

\begin{proof}
To obtain the derivative formulas, let us consider the derivative of the orthogonal relation \eqref{t-or}, from which we have
\begin{align}\label{t-innerp}
\begin{split}
\p_{t_\theta}H_n(\mt)\delta_{n,m}=&\langle \p_{t_\theta}P_n(x;\mt),Q_m(x;\mt)\rangle_\theta
+\langle P_n(x;\mt),\p_{t_\theta}Q_m(x;\mt)\rangle_\theta\\
&+\langle x^\theta P_n(x;\mt),Q_m(x;\mt)\rangle_\theta.
\end{split}
\end{align}
By assuming that $m<n$, the above equation implies
\begin{align}\label{de1}
\langle \p_{t_\theta}P_n(x;\mt),Q_m(x;\mt)\rangle_\theta=-\langle  P_n(x;\mt),xQ_m(x;\mt)\rangle_\theta.
\end{align}
Since $\p_{t_\theta}P_n(x;\mt)$ is a polynomial of degree $n-1$, it could be expanded as a left linear combination of the orthogonal basis $\{P_k(x;\mt)\}_{k=0}^{n-1}$ and
\begin{align*}
\p_{t_\theta}P_n(x;\mt)=\sum_{k=0}^{n-1}\gamma_{n,k}P_k(x;\mt).
\end{align*}
By taking $m=0, 1, \dots, n-1$ in \eqref{de1}, we obtain \eqref{t-p}. On the other hand, by assuming that $n<m$, then
\begin{align*}
\langle P_n(x;\mt), \p_{t_\theta}Q_m(x;\mt)\rangle_\theta=-\langle x^\theta P_n(x;\mt),Q_m(x;\mt)\rangle_\theta,
\end{align*}
and in a similar manner, we can obtain \eqref{t-q}.
\end{proof}

With matrix-form notation, equations \eqref{t-p} and \eqref{t-q} can  alternatively be  written as 
\begin{align*}
\p_{t_\theta}\Phi=-L_{<0}\Phi,\quad \p_{t_\theta}\Psi=-M_{<0}\Psi,
\end{align*}
where $L_{<0}$ means the strictly lower triangular part of the block matrix.
Thus the compatibility conditions $x^\theta\p_{t_\theta}\Phi=\p_{t_\theta}(x^\theta\Phi)$ and $x\p_{t_\theta}\Psi=\p_{t_\theta}(x\Psi)$ result in integrable lattices
\begin{align*}
\p_{t_\theta}L=[L,L_{<0}],\quad \p_{t_\theta}M=[M,M_{<0}].
\end{align*}
If these equations are written into explicit elements, then we have
\begin{align}\label{nc-bm}
\left\{\begin{array}{ll}
\p_{t_\theta}\alpha_{n,n+\theta-1}=\alpha_{n+\theta,n+\theta-1}-\alpha_{n,n-1},\\
\p_{t_\theta}\alpha_{n,i}=\alpha_{n, i+1}\alpha_{i+1,i}-\alpha_{n,n-1}\alpha_{n-1,i},& i=n-1,  \dots, n+\theta-2,
\end{array}
\right.
\end{align}
and 
\begin{align}\label{nc-kl}
\left\{\begin{array}{ll}
\p_{t_\theta}\beta_{n,n}=\beta_{n+1,n}-\beta_{n,n-1},\\
\p_{t_\theta}\beta_{n,j}=\beta_{n+1, j}-\beta_{n,j-1}+\beta_{n,n}\beta_{n,j}-\beta_{n,j}\beta_{j,j}, & j=n-\theta, \dots, n-1,
\end{array}
\right.
\end{align}
where $\beta_{n,n-\theta-1}$ is defined to be zero.
Equations \eqref{nc-bm} and \eqref{nc-kl} are non-commutative generalizations of the Toda-type equations. 
In this paper, we call them non-commutative hungry Toda lattices, as the corresponding commutative cases are hungry Toda lattices, which are special situations of Kostant-Toda hierarchy \cite{shinjo20}.
Moreover, in the commutative case, these equations were studied from different perspectives, such as the discrete Lax formalism, the $r$-matrix approach and Hamiltonian structures \cite{blaszak94,kuperschmidt85}. 
In particular, we can also refer to equation \eqref{nc-bm} as the non-commutative Blaszak-Marciniak equation since Blaszak and Marciniak studied the commutative case \cite{blaszak94} via the r-matrix approach, and to equation \eqref{nc-kl} as the non-commutative Kuperschmidt lattice since the Hamiltonian structure of the commutative equation was considered by Kuperschmidt \cite{kuperschmidt85}. \\

\begin{definition}\label{wq}
Given that $\{m_i\}_{i=1}^N$ is dependent on a time parameter $t$, if the $k$-th derivative of $m_i$ is denoted by $m_i^{(k)}$,  then the corresponding Wronski quasi-determinant is defined by
\begin{align*}
\left|\begin{array}{cccc}
m_1^{(0)}&m_1^{(1)}&\cdots&m_{1}^{(N-1)}\\
m_2^{(0)}&m_2^{(1)}&\cdots&m_{2}^{(N-1)}\\
\vdots&\vdots&&\vdots\\
m_N^{(0)}&m_N^{(1)}&\cdots&\boxed{m_N^{(N-1)}}
\end{array}
\right|.
\end{align*}  
\end{definition}
\begin{remark}
With time parameters involved, the normalization factors $\{H_n\}_{n\in\mathbb{N}}$ in \eqref{nf} are  Wronski quasi-determinants. Moreover, these Wronski quasi-determinants play the role of non-commutative $\tau$-functions for non-commutative integrable systems \eqref{nc-bm} and \eqref{nc-kl}.
\end{remark}

\subsection{A quasi-determinant solution to non-commutative Blaszak-Marciniak three-field equations}\label{sec2.3}
In this part, we intend to make use of quasi-Wronski techniques to show that the  Blaszak-Marciniak equation admits quasi-determinant solutions.  
The first non-trivial example in the Blaszak-Marciniak equation is when $\theta=1$, which gives rise to the non-commutative Toda equation
\begin{align*}
\p_{t_1} a_n=b_{n+1}-b_n,\quad \p_{t_1}b_n=a_nb_n-b_n a_{n-1}
\end{align*}
with $a_n=\left(\p_{t_1}H_n\right) H_n^{-1}$ and $b_n=H_n H_{n-1}^{-1}$, where $H_n$ is the corresponding Hankel quasi-determinant. Verifications of the non-commutative Toda lattice have been exhibited in \cite{gelfand05,li08,retakh10,li20} and so we omit the details here. 

Therefore, we pay attention to the second non-trivial example (i.e. $\theta=2$ case), the non-commutative Blaszak-Marciniak three-field equations. The commutative Blaszak-Marciniak three-field equations attracted much attention. For example, its B\"acklund transformation and superposition formula was given in \cite{hu98}, its algebro-geometric solution was given in \cite{geng17}, and its connection with moving frame was given in \cite{wang21}. 
For the non-commutative case, the Hamiltonian structure and recursion operator were recently given in \cite{casati21}. 
To demonstrate the quasi-determinant solution of the non-commutative Blaszak-Marciniak three-field equations, we have the following theorem.
\begin{theorem}
The non-commutative Blaszak-Marciniak three-field equations 
\begin{subnumcases} 
{\label{nc-bm-2}}
\p_{t_2}\alpha_{n,n+1}=\alpha_{n+2,n+1}-\alpha_{n,n-1},\label{nc-bm-2-1}\\
\p_{t_2}\alpha_{n,n}=\alpha_{n, n+1}\alpha_{n+1,n}-\alpha_{n,n-1}\alpha_{n-1,n},\label{nc-bm-2-2}\\
\p_{t_2}\alpha_{n,n-1}=\alpha_{n, n}\alpha_{n,n-1}-\alpha_{n,n-1}\alpha_{n-1,n-1},\label{nc-bm-2-3}
\end{subnumcases}
admit following solutions
\begin{align*}
&\alpha_{n,n-1}=H_nH_{n-1}^{-1},\quad \alpha_{n,n}=(Z_{n,n}+H_n\eta_{n,n-1})H_n^{-1},\\
&\alpha_{n,n+1}=\xi_{n,n-1}-\xi_{n+2,n+1}=\left(
Z_{n,n+1}+Z_{n,n}\eta_{n+1,n}+H_n\eta_{n+1,n-1}
\right)H_{n+1}^{-1}
\end{align*}
in which
\begin{align*}
\xi_{n+1,n}=\!\!\left|\begin{array}{cccc}
\!\!m_{0}&\cdots&\!\!\!\!m_{2n}&\!\!\!\!0\\
\vdots&&\vdots&\vdots\\
\!\!m_{n}&\cdots&\!\!\!\!m_{3n}&\!\!\!\!\mathbb{I}_p\\
\!\!m_{n+1}&\!\!\!\!\cdots&m_{3n+1}&\!\!\!\!\boxed{0}\end{array}\right|,
\quad
\eta_{n+1,j}&=\!\!\left|\begin{array}{cccc}
\!\!m_{0}&\!\!\!\!\cdots&\!\!\!\!m_{2n}&\!\!\!\!m_{2n+2}\\
\vdots&&\!\!\!\!\vdots&\vdots\\
\!\!m_{n}&\!\!\!\!\cdots&\!\!\!\!m_{3n}&\!\!\!\!m_{3n+2}\\
\!\!&\!\!\!\!e_{j+1}&&\!\!\!\!\boxed{0}\end{array}\right|,
\end{align*}
and 
\begin{align*}
H_{n}=\!\!\left|\begin{array}{cccc}
\!\!m_{0}&\cdots&\!\!\!\!m_{2n-2}&\!\!\!\!m_{2n}\\
\vdots&&\vdots&\vdots\\
\!\!m_{n-1}&\cdots&\!\!\!\!m_{3n-3}&\!\!\!\!m_{3n-1}\\
\!\!m_{n}&\!\!\!\!\cdots&m_{3n-2}&\!\!\!\!\boxed{m_{3n}}\end{array}\right|,\quad 
Z_{n,j}=\!\!\left|\begin{array}{cccc}
\!\!m_{0}&\cdots&\!\!\!\!m_{2n-2}&\!\!\!\!m_{2j+2}\\
\vdots&&\vdots&\vdots\\
\!\!m_{n-1}&\cdots&\!\!\!\!m_{3n-3}&\!\!\!\!m_{n+1+2j}\\
\!\!m_{n}&\!\!\!\!\cdots&m_{3n-2}&\!\!\!\!\boxed{m_{n+2+2j}}\end{array}\right|,
\end{align*}
under time evolution $\p_{t_2}m_i=m_{i+2}$.
\end{theorem}
\begin{remark}
Below, we show that $\xi_{n+1,n}$, $\eta_{n+1,j}$ and $Z_{n,j}$ are quantities related to the Wronski quasi-determinant $H_n$ and its derivatives. Therefore, we say that the non-commutative Blaszak-Marciniak three field equation could be simply expressed in terms of the Wronski quasi-determinant.
\end{remark}

To prove this theorem, firstly, we notice the following homological relation.
\begin{lemma}
It holds that
\begin{align}\label{eta-z}
Z_{n,n}=-H_n\eta_{n+1,n}
\end{align}
\end{lemma}
\begin{proof}
By using the non-commutative Jacobi identity \eqref{ncj1}, we have
\begin{align*}
\eta_{n+1,n}&=\left|\begin{array}{ccccc}
m_0&\cdots&m_{2n-2}&m_{2n}&m_{2n+2}\\
\vdots&&\vdots&\vdots&\vdots\\
m_{n-1}&\cdots&m_{3n-3}&m_{3n-1}&m_{3n+1}\\
m_n&\cdots&m_{3n-2}&m_{3n}&m_{3n+2}\\
0&\cdots&0&\mathbb{I}_p&\boxed{0}
\end{array}
\right|\\
&=-\left|\begin{array}{cccc}
m_0&\cdots&m_{2n-2}&m_{2n}\\
\vdots&&\vdots&\vdots\\
m_{n-1}&\cdots&m_{3n-3}&m_{3n-1}\\
m_n&\cdots&m_{3n-2}&\boxed{m_{3n}}
\end{array}
\right|^{-1}
\left|\begin{array}{cccc}
m_0&\cdots&m_{2n-2}&m_{2n+2}\\
\vdots&&\vdots&\vdots\\
m_{n-1}&\cdots&m_{3n-3}&m_{3n+1}\\
m_n&\cdots&m_{3n}&\boxed{m_{3n+2}}
\end{array}
\right|\\
&=-H_n^{-1}Z_{n,n}.
\end{align*}
\end{proof}

Moreover, we have the following derivative formulas for $\xi_{n,n-1}$. 

\begin{lemma}
Regarding with the derivative of $\xi_{n,n-1}$, it holds that
\begin{align}\label{pt-xi}
\p_{t_2}\xi_{n,n-1}=-H_{n}H_{n-1}^{-1}.
\end{align}
\end{lemma}
\begin{proof}
According to the derivative formula for quasi-determinants, we obtain 
\begin{align*}
\p_{t_2}\xi_{n,n-1}&=\left|\begin{array}{cccc}
m_0&\cdots&m_{2n-4}&0\\
\vdots&&\vdots&\vdots\\
m_{n-2}&\cdots&m_{3n-6}&\mathbb{I}_p\\
m_{n+1}&\cdots&m_{3n-2}&\boxed{0}
\end{array}
\right|\\
&+\sum_{k=1}^{n-1} \left|\begin{array}{cccc}
m_0&\cdots&m_{2n-4}&m_{2k}\\
\vdots&&\vdots&\vdots\\
m_{n-2}&\cdots&m_{3n-6}&m_{2k+n-2}\\
m_{n-1}&\cdots&m_{3n-5}&\boxed{0}
\end{array}
\right|\cdot\left|\begin{array}{cccc}
m_0&\cdots&m_{2n-4}&0\\
\vdots&&\vdots&\vdots\\
m_{n-2}&\cdots&m_{3n-6}&\mathbb{I}_p\\
&e_k&&\boxed{0}
\end{array}
\right|,
\end{align*}
where $e_k$ is the $k$-th unit vector given in \eqref{ue}. Moreover, by noting that 
\begin{align}\label{eq0}
\left|\begin{array}{cccc}
m_0&\cdots&m_{2n-4}&m_{2k}\\
\vdots&&\vdots&\vdots\\
m_{n-2}&\cdots&m_{3n-6}&m_{2k+n-2}\\
m_{n-1}&\cdots&m_{3n-5}&\boxed{0}
\end{array}\right|=\left\{\begin{array}{ll}
-m_{2k+n-1},&k=1,\cdots,n-2\\
-m_{3n-3}+H_{n-1},&k=n-1
\end{array}
\right.
\end{align}
we get
\begin{align*}
\p_{t_2}\xi_{n,n-1}=H_{n}\left|\begin{array}{cccc}
m_0&\cdots&m_{2n-2}&0\\
\vdots&&\vdots&\vdots\\
m_{n-1}&\cdots&m_{3n-3}&\mathbb{I}_p\\
0&\cdots&\mathbb{I}_p&\boxed{0}
\end{array}
\right|=-H_nH_{n-1}^{-1},
\end{align*}
where a non-commutative Jacobi identity \eqref{ncj1} is applied to the last step.
\end{proof}
Therefore, equation \eqref{nc-bm-2-1} can be verified directly from \eqref{pt-xi}.

\begin{lemma}
Regarding the derivative of $H_{n}$, it holds that
\begin{align}\label{pt-H}
\p_{t_2}H_{n}=Z_{n,n}+H_n\eta_{n,n-1}.
\end{align}
\end{lemma}
\begin{proof}
The derivative formula for $H_n$ is given by
\begin{align*}
\!\!\p_{t_2}H_{n}\!\!
&=\left|\begin{array}{cccc}
m_0&\cdots&m_{2n-2}&m_{2n+2}\\
m_1&\cdots&m_{2n-1}&m_{2n+3}\\
\vdots&&\vdots&\vdots\\
m_{n-1}&\cdots&m_{3n-3}&m_{3n+1}\\
m_n&\cdots&m_{3n-2}&\boxed{m_{3n+2}}
\end{array}\right|
+\left|\begin{array}{cccc}
m_0&\cdots&m_{2n-2}&m_{2n}\\
m_1&\cdots&m_{2n-1}&m_{2n+1}\\
\vdots&\vdots&&\vdots\\
m_{n-1}&\cdots&m_{3n-3}&m_{3n-1}\\
m_{n+2}&\cdots&m_{3n}&\boxed{0}
\end{array}
\right|\\
&+\sum_{j=1}^n\left|\begin{array}{cccc}
m_0&\cdots&m_{2n-2}&m_{2j}\\
m_1&\cdots&m_{2n-1}&m_{2j+1}\\
\vdots&\vdots&&\vdots\\
m_{n-1}&\cdots&m_{3n-3}&m_{n-1+2j}\\
m_{n}&\cdots&m_{3n-2}&\boxed{0}
\end{array}\right|\cdot
\left|\begin{array}{cccc}
m_0&\cdots&m_{2n-2}&m_{2n}\\
m_1&\cdots&m_{2n-1}&m_{2n+1}\\
\vdots&&\vdots&\vdots\\
m_{n-1}&\cdots&m_{3n-3}&m_{3n-1}\\
&e_j&&\boxed{0}
\end{array}
\right|.
\end{align*}
By applying the equation \eqref{eq0}, we know that only $j=n$ is left in the sum, and \eqref{pt-H} is verified.
\end{proof}
\begin{remark}
This lemma indicates that $\alpha_{n,n}=\left(\p_{t_2}H_n\right) H_n^{-1}$.
\end{remark}
Therefore, by substituting formulas \eqref{eta-z} and \eqref{pt-H} into \eqref{nc-bm-2-3}, we know that the third equation of the non-commutative Blaszak-Marciniak three-field equations is verified. To verify \eqref{nc-bm-2-2}, a derivative formula for $\eta_{n+1,n}$ is needed. 

\begin{lemma}
Regarding with the derivative of $\eta_{n+1,n}$, we have that
\begin{align}\label{ddd}
\p_{t_2}\eta_{n+1,n}=\eta_{n+1,n}\eta_{n+1,n}-H_n^{-1}Z_{n,n+1}-\eta_{n+1,n-1}.
\end{align}
\end{lemma}
\begin{proof}
Taking the derivative of $\eta_{n+1,n}$, one obtains
\begin{align*}
\p_{t_2}\eta_{n+1,n}&=\left|\begin{array}{cccc}
m_0&\cdots&m_{2n}&m_{2n+4}\\
m_1&\cdots&m_{2n+1}&m_{2n+5}\\
\vdots&&\vdots&\vdots\\
m_n&\cdots&m_{3n}&m_{3n+4}\\
0&\cdots&\mathbb{I}_p&\boxed{0}\end{array}
\right|\\&+\sum_{j=1}^{n+1}\left|\begin{array}{cccc}
m_0&\cdots&m_{2n}&m_{2j}\\
m_1&\cdots&m_{2n+1}&m_{2j+1}\\
\vdots&&\vdots&\vdots\\
m_n&\cdots&m_{3n}&m_{n+2j}\\
0&\cdots&\mathbb{I}_p&\boxed{0}
\end{array}
\right|\cdot\left|\begin{array}{cccc}
m_0&\cdots&m_{2n}&m_{2n+2}\\
m_1&\cdots&m_{2n+1}&m_{2n+3}\\
\vdots&&\vdots&\vdots\\
m_n&\cdots&m_{3n}&m_{3n+2}\\
&e_j&&\boxed{0}
\end{array}
\right|.
\end{align*}
Noting that
\begin{align*}
\left|\begin{array}{cccc}
m_0&\cdots&m_{2n}&m_{2j}\\
m_1&\cdots&m_{2n+1}&m_{2j+1}\\
\vdots&&\vdots&\vdots\\
m_n&\cdots&m_{3n}&m_{n+2j}\\
0&\cdots&\mathbb{I}_p&\boxed{0}
\end{array}
\right|=\left\{\begin{array}{ll}
0&j=1,\cdots,n-1,\\
-\mathbb{I}_p&j=n,\\
\eta_{n+1,n}&j=n+1,
\end{array}
\right.
\end{align*}
we have a simplified equation
\begin{align*}
\p_{t_2}\eta_{n+1,n}=\left|\begin{array}{cccc}
m_0&\cdots&m_{2n}&m_{2n+4}\\
m_1&\cdots&m_{2n+1}&m_{2n+5}\\
\vdots&&\vdots&\vdots\\
m_n&\cdots&m_{3n}&m_{3n+4}\\
0&\cdots&\mathbb{I}_p&\boxed{0}\end{array}
\right|+\eta_{n+1,n}\eta_{n+1,n}-\eta_{n+1,n-1}.
\end{align*}
The application of non-commutative Jacobi identity \eqref{ncj1} to $(n+1,n+2)$-rows and $(n+1,n+2)$-columns to the above quasi-determinant gives the result of \eqref{ddd}.
\end{proof}
Therefore, by using \eqref{ddd},  \eqref{nc-bm-2-2} can be directly verified.

\section{A non-commutative generalization of the bigraded Toda lattice}\label{sec3}
In this section, we propose to generalize the non-commutative Blaszak-Marciniak lattice and Kuperschmidt lattice to the extended bigraded Toda lattice case, which in the commutative case, admits the Lax operator \cite{carlet04,carlet06,dubrovin04}
\begin{align}\label{bi-lax}
\mathcal{L}(a,b)=\Lambda^a+c_1\Lambda^{a-1}+\cdots+c_{a+b}\Lambda^{-b},\quad a,b\in \mathbb{Z}_+.
\end{align}
To construct this Lax operator by using bi-orthogonal polynomials, we need to consider the parameter $\theta\in\mathbb{Q}_+$ instead of $\theta\in\mathbb{Z}_+$ in bilinear form \eqref{innerproduct}.
Let us first consider recurrence relations for $\{P_n(x),Q_n(x)\}_{n\in\mathbb{N}}$ when $\theta=b/a$, where $a,b$ are positive integers. It should be noted that to uniquely determine the value of $\theta$, we need to require that $a$ and $b$ should be co-prime. However, the recurrence relation is dependent on the values of $a$ and $b$ rather than that of $\theta$, as we demonstrate below.

\begin{proposition}
Given
\begin{align*}
\theta=\frac{b}{a},\quad a,\,b\in\mathbb{Z}_+,
\end{align*}
 the so-called $(a,b)$-graded matrix-valued bi-orthogonal polynomials $\{P_n(x), Q_n(x)\}_{n\in\mathbb{N}}$ satisfying the bi-orthogonal relation under \eqref{or} have the following recurrence relations
\begin{subequations}
\begin{align}
x^b P_n(x)=P_{n+b}(x)+\sum^{n+b-1}_{j=n-a}\alpha_{n, j}P_j(x),\label{rational-re-1}\\
x^a Q_n(x)=Q_{n+a}(x)+\sum^{n+a-1}_{j=n-b}\beta_{n, j}Q_j(x).\label{rational-re-2}
\end{align}
\end{subequations}
Moreover, coefficients $\alpha_{n, j}$ and $\beta_{n, j}$ can be written as
\begin{align*}
\alpha_{n, j}=\left(\mathcal{Z}_{n,j}+\sum^{j-1}_{i=n-a}\mathcal{Z}_{n,i}\eta_{j,i}\right)H_j\!^{-1},\\
\beta^T_{n, j}=H_j\!^{-1}\left(\mathcal{Y}_{n,j}+\sum^{j-1}_{i=n-b}\xi_{j,i}\mathcal{Y}_{n,i}\right),
\end{align*}
where $H_n$ has the same expression as in \eqref{nf}, and $\mathcal{Z}_{n,j}$ and $\mathcal{Y}_{n,j}$ can be written as quasi-determinants
\begin{align*}
\mathcal{Z}_{n,j}={\langle P_n(x), x^{a+j}\mathbb{I}_p\rangle}_\theta
=\left|
\begin{array}{ccccc}
m_{0}&m_{\theta}&\cdots&m_{(n-1)\theta}&m_{(a+j)\theta}\\
\vdots&\vdots& &\vdots&\vdots\\
m_{n-1}&m_{n-1+\theta}&\cdots&m_{n-1+(n-1)\theta}&m_{n-1+(a+j)\theta}\\
m_{n}&m_{n+\theta}&\cdots&m_{n+(n-1)\theta}&\boxed{m_{n+(a+j)\theta}}
\end{array}
\right|,
\end{align*}
and
\begin{align*}
\mathcal{Y}_{n, j}= \langle {x^{b+j}\mathbb{I}_p,Q_n(x)\rangle}_\theta
=\left|
\begin{array}{ccccc}
m_{0}&m_{\theta}&\cdots&m_{(n-1)\theta}&m_{n\theta}\\
\vdots&\vdots& &\vdots&\vdots\\
m_{n-1}&m_{n-1+\theta}&\cdots&m_{n-1+(n-1)\theta}&m_{n-1+n\theta}\\
m_{b+j}&m_{b+j+\theta}&\cdots&m_{b+j+(n-1)\theta}&\boxed{m_{b+j+n\theta}}
\end{array}
\right|.
\end{align*}
\end{proposition}
\begin{proof}
The proof of this proposition is similar to the ones for $\theta\in\mathbb{Z}_+$ by rewriting the quasi-symmetry property
\begin{align}\label{qscon}
\langle x^b P_n(x), Q_m(x)\rangle_\theta = \langle P_n(x),x^a Q_m(x)\rangle_\theta,\quad \theta=\frac{b}{a}\in\mathbb{Q}_+.
\end{align}
\end{proof}
If we denote 
\begin{align*}
\mathcal{L}(b,a)=\Lambda^b+\alpha_{b-1}\Lambda^{b-1}+\cdots+\alpha_{-a}\Lambda^{-a},\quad \alpha_i=(\alpha_{0,i},\alpha_{1,i+1},\cdots),\\
\mathcal{M}(a,b)=\Lambda^a+\beta_{a-1}\Lambda^{a-1}+\cdots+\beta_{-b}\Lambda^{-b},\quad \beta_i=(\beta_{i,0},\beta_{i+1,1},\cdots),
\end{align*}
then \eqref{rational-re-1} and \eqref{rational-re-2} can be equivalently written as
\begin{align}\label{spec1}
x^b\Phi=\mathcal{L}(b,a)\Phi,\quad x^a\Psi=\mathcal{M}(a,b)\Psi
\end{align}
where $\Phi$ and $\Psi$ are wave functions defined in \eqref{wave1}.

Similarly, we can introduce infinitely many time flows $\{t_1,t_2,\cdots\}$ such that the weight function is dependent on the time flows as in \eqref{td}, then the evolution of the wave function can be stated as follows.

\begin{proposition}\label{proprationalt}
\begin{subequations}
The matrix-valued bi-orthogonal polynomials $\{P_n(x;\mt)\}_{n\in\mathbb{N}}$  and $\{Q_n(x;\mt)\}_{n\in\mathbb{N}}$  satisfy the evolution equations
\begin{align}\label{rational-t-p}
\p_{t_b}P_n(x;\mt)=-\sum^{n-1}_{i=n-a}\alpha_{n,i}P_{i}(x;\mt),
\end{align}
\begin{align}\label{rational-t-q}
\p_{t_b}Q_n(x;\mt)=-\sum^{n-1}_{i=n-b}\beta_{n,i}Q_{i}(x;\mt),
\end{align}
where $\{\alpha_{n,i}\}_{i=n-a}^{n-1}$ and $\{\beta_{n,i}\}_{i=n-b}^{n-1}$ are the recurrence coefficients in \eqref{rational-re-1} and  \eqref{rational-re-2} respectively.
\end{subequations}
\end{proposition}
\begin{remark}
The proofs of \eqref{rational-t-p} and \eqref{rational-t-q} are similar to those in Proposition \ref{proptp}. However, one should notice why we consider the $t_b$-flow. When we take the derivative of the orthogonal relation \eqref{or} with respect to the $t_j$-flow, we have
\begin{align*}
\langle \p_{t_j}P_n(x;\mt),Q_m(x;\mt)\rangle_\theta=-\langle x^jP_n(x;\mt), Q_m(x;\mt)\rangle_\theta
\end{align*}
similar to \eqref{de1}. Then the evaluation of $\langle x^j P_n(x;\mt), Q_m(x;\mt)\rangle_\theta$ determines the coefficients when expanding the derivative of polynomials. Therefore, to truncate the derivative formula, one needs to take $j$ as a multiplier of $b$ and makes use of the quasi-symmetry property. Proposition \ref{proprationalt} gives the simplest case.
However, it doesn't include all cases. If we consider a continuous independent variable $x$, then fractional powers and logarithm of the Lax operator $\mathcal{L}$ will be used. Further discussions with Frobenius manifolds have been given in \cite{carlet06} in commutative case, and its connection with Hurwitz numbers has been discussed in \cite{takasaki18}.
\end{remark}
Therefore, if one rewrites \eqref{rational-t-p} and \eqref{rational-t-q} as a matrix form
\begin{align*}
\p_{t_b}\Phi=-\mathcal{L}(b,a)_{<0}\Phi,\quad \p_{t_b}\Psi=-\mathcal{M}(a,b)_{<0}\Psi,
\end{align*}
then compatibility conditions of \eqref{rational-re-1} and \eqref{rational-t-p} give rise to the non-commutative $(b,a)$-graded Toda lattice
\begin{align*}
\p_{t_b}\mathcal{L}(b,a)=[\mathcal{L}(b,a),\mathcal{L}(b,a)_{<0}],
\end{align*}
while \eqref{rational-re-2} and \eqref{rational-t-q} lead to the non-commutative $(a,b)$-graded Toda lattice
\begin{align*}
\p_{t_b}\mathcal{M}(a,b)=[\mathcal{M}(a,b),\mathcal{M}(a,b)_{<0}].
\end{align*}
Moreover, we can generalize the above procedure to higher order time flows, namely $t_{kb}$-flow ($k=1,2,\cdots$), which generates a non-commutative bigraded Toda hierarchy.
\begin{proposition}\label{prop4}
\begin{subequations}
The derivative of $P_n(x;\mt)$ with respect to the $t_{kb}$-flow is given by the formula
\begin{align}
\p_{t_{kb}}\Phi=-\left(\mathcal{L}^k(b,a)\right)_{<0}\Phi,\label{tkb1}
\end{align}
and that of $Q_n(x;\mt)$ is given by
\begin{align}
\p_{t_{kb}}\Psi=-\left(\mathcal{M}^k(a,b)\right)_{<0}\Psi.\label{tkb2}
\end{align}
\end{subequations}
\end{proposition}

\begin{proof}
Let us introduce the notation
\begin{align*}
\left(\langle\Phi, \Psi\rangle_\theta\right):=\left(\begin{array}{ccc}
\langle P_0(x;\mt),Q_0(x;\mt)\rangle_\theta&\langle P_0(x;\mt),Q_1(x;\mt)\rangle_\theta&\cdots\\
\langle P_1(x;\mt),Q_0(x;\mt)\rangle_\theta&\langle P_1(x;\mt),Q_1(x;\mt)\rangle_\theta&\cdots\\
\vdots&\vdots&\ddots
\end{array}
\right),
\end{align*}
Using the  orthogonal relations \eqref{or}, we have $\left(\langle\Phi, \Psi\rangle_\theta\right)=\mathcal{H}$, which is a block diagonal matrix with elements $(H_0, H_1, \cdots)$. From the spectral problem \eqref{spec1}, we have $x^{kb}\Phi=\mathcal{L}^{k}(b,a)\Phi$. Therefore 
 \begin{align*}
 (\langle x^{kb}\Phi, \Psi\rangle_\theta)=(\langle\mathcal{L}^{k}(b,a)\Phi, \Psi\rangle_\theta)=\mathcal{L}^{k}(b,a)(\langle\Phi, \Psi\rangle_\theta)=\mathcal{L}^{k}(b,a)\mathcal{H}.
 \end{align*}
On the other hand, by taking the derivative for the orthogonal relation, we get
\begin{align*}
\p_{t_{kb}}\mathcal{H}
&=\left(\langle\p_{t_{kb}}\Phi, \Psi\rangle_\theta\right)+\left(\langle\Phi,\p_{t_{kb}} \Psi\rangle_\theta\right)+\left(\langle x^{kb}\Phi, \Psi\rangle_\theta\right)\\
&=\left(\langle\p_{t_{kb}}\Phi, \Psi\rangle_\theta\right)+\left(\langle\Phi,\p_{t_{kb}} \Psi\rangle_\theta\right)+\mathcal{L}^k (b,a)\mathcal{H},
\end{align*}
where $\left(\langle \p_{t_{kb}}\Phi, \Psi\rangle_\theta\right)$ is a  strictly block lower triangular matrix and 
$\left(\langle \Phi,\p_{t_{kb}} \Psi\rangle_\theta\right)$ is strictly block upper triangular. Thus, according to the decomposition of matrix, we have
\begin{align*}
\left(
\langle \p_{t_{kb}}\Phi,\Psi\rangle_\theta
\right)=-\left(\mathcal{L}^k(b,a)\mathcal{H}\right)_{<0}=-\left(\mathcal{L}^k(b,a)\right)_{<0}\mathcal{H}.
\end{align*}
Moreover, by the Riesz representation theorem \cite{reed80}, we get
\begin{align*}
\p_{t_{kb}}\Phi=-\left(\mathcal{L}^k(a,b)\right)_{<0}\Phi.
\end{align*}
The derivative of $Q_n(x;\mt)$ can be similarly verified.
\end{proof} 
Therefore, the compatibility conditions of \eqref{tkb1}, \eqref{tkb2} and \eqref{spec1} result in the non-commutative bigraded Toda hierarchy
\begin{align*}
\p_{t_{kb}}\mathcal{L}(b,a)&=[\mathcal{L}(b,a),\mathcal{L}^k(b,a)_{<0}],\\
\p_{t_{kb}}\mathcal{M}(a,b)&=[\mathcal{M}(a,b),\mathcal{M}^k(a,b)_{<0}].
\end{align*}

\section{Moment reduction and B\"acklund transformation}\label{sec4}

B\"acklund transformations are an important tool in soliton theory. They relate different solutions to the same equation. It is well known that in the commutative case, the Lotka-Volterra equation is the B\"acklund transformation of the Toda equation, see e.g. \cite{chu08,hirota78}. Such a relation has been generalized to the hungry Toda case (i.e. commutative Blaszak-Marciniak/Kuperschmidt lattice hierarchy) and the hungry Lotka-Volterra hierarchy (or so-called Itoh-Narita-Bogoyavlensky lattice hierarchy) \cite{fukuda11,shinjo22}.

Therefore, it is of interest to understand B\"acklund transformations from the orthogonal polynomials perspective, as well as to obtain a non-commutative generalization of the Itoh-Narita-Bogoyavlensky lattice hierarchy. To this end, we start from the Wronski quasi-determinant solutions of the Blaszak-Marciniak/Kuperschmidt lattice hierarchy, and introduce the moment reduction approach.

\subsection{A B\"acklund transformation for the Blaszak-Marciniak lattice}\label{sec4.1}
We first consider the B\"acklund transformation for $\theta\in\mathbb{Z}_+$ case. For $\theta=1$, the B\"acklund transformation between the non-commutative Toda and Lotka-Volterra has been revealed in \cite{li20}. 

Let us first claim the following proposition.
\begin{proposition}\label{propbt}
If 
\begin{align*}
H_n=\left|\begin{array}{cccc}
m_0&m_\theta&\cdots&m_{n\theta}\\
m_1&m_{\theta+1}&\cdots&m_{n\theta+1}\\
\vdots&\vdots&&\vdots\\
m_n&m_{\theta+n}&\cdots&\boxed{m_{n\theta+n}}
\end{array}
\right|,\quad n=0,1,\cdots
\end{align*}
are solutions to \eqref{nc-bm} and \eqref{nc-kl} under proper evolutions $\p_t m_i=m_{i+\theta}$, then 
\begin{align*}
H_n^{(\ell)}=\left|\begin{array}{cccc}
m_\ell&m_{\ell+\theta}&\cdots&m_{\ell+n\theta}\\
m_{\ell+1}&m_{\ell+\theta+1}&\cdots&m_{\ell+n\theta+1}\\
\vdots&\vdots&&\vdots\\
m_{\ell+n}&m_{\ell+\theta+n}&\cdots&\boxed{m_{\ell+n\theta+n}}
\end{array}
\right|,\quad n=0,1,\cdots
\end{align*}
are still solutions to \eqref{nc-bm} and \eqref{nc-kl} under the same evolutions.
\end{proposition}

\begin{remark}
This proposition is an alternative expression of the Wronski property. It is known that if the seed functions are properly chosen such that an equation admits a proper Wronskian-type solution, then the expression of solutions is independent of the choice of seed functions \cite[\S 4.3.1]{hirota04}. Such an idea also appears in the construction of non-commutative integrable systems by using the Marchenko lemma \cite{etingof98}.
\end{remark}

Based on this proposition, there are two questions remaining if we want to construct a B\"acklund transformation for equations \eqref{nc-bm} and \eqref{nc-kl}. One is how to construct $H_n^{(\ell)}$, and another is how to find a possible connection equation for those $H_n^{(\ell)}$. 
To construct $H_n^{(\ell)}$, a moment reduction method is needed.
\begin{proposition}\label{prop5.2}
If we take the moment reduction 
\begin{align}\label{mred}
\langle x^i\mathbb{I}_p,x^j\mathbb{I}_p\rangle_\theta=m_{i+j\theta}:=\left\{\begin{array}{lc}
d_{\frac{i+j\theta}{\theta+1}},& i+j\theta \mod \theta+1=0,\\
0,&i+j\theta\mod \theta+1\ne0,
\end{array}
\right.
\end{align}
then $\{H_{n(\theta+1)+\ell}\}_{\ell=0,1,\cdots,\theta,\,n\in\mathbb{N}}$ are graded and have expressions of Wronski quasi-determinants
\begin{align*}
H_{n(\theta+1)+\ell}=\left|\begin{array}{cccc}
d_\ell&d_{\ell+\theta}&\cdots&d_{\ell+n\theta}\\
d_{\ell+1}&d_{\ell+1+\theta}&\cdots&d_{\ell+1+n\theta}\\
\vdots&\vdots&&\vdots\\
d_{\ell+n}&d_{\ell+n+\theta}&\cdots&\boxed{d_{\ell+n+n\theta}}
\end{array}
\right|,\quad \ell=0,1,\cdots,\theta,\, n\in\mathbb{N}
\end{align*}
\end{proposition}
\begin{proof}
We prove the quasi-determinant expressions for $H_{n(\theta+1)}$, and the others can be similarly verified.
Expanding $H_{n(\theta+1)}$ by definition, we have 
\begin{align*}
H_{n(\theta+1)}=d_{n+n\theta}-(v_0,v_1,\cdots,v_{n-1},\tilde{v}_{n})\left(
\begin{array}{cccc}
M_0&M_\theta&\cdots&\hat{M}_{n\theta}\\
M_1&M_{\theta+1}&\cdots&\hat{M}_{n\theta+1}\\
\vdots&\vdots&&\vdots\\
\tilde{M}_n&\tilde{M}_{\theta+n}&\cdots&\hat{\tilde{M}}_{n\theta+n}
\end{array}
\right)^{-1}\left(\begin{array}{c}
w_0\\w_1\\\vdots\\\tilde{w}_{n}
\end{array}
\right),
\end{align*}
where $\tilde{v}_n$ (resp. $\tilde{w}_n$) is a $\theta$ row (resp. column) zero vector, and $v_j$ (resp. $w_j$) are $\theta+1$ row (resp. column) vectors of the form
\begin{align*}
v_j=\left(d_{n+j\theta},0,\cdots,0
\right),\quad w_j=\left(\begin{array}{c}
d_{n\theta+j}\\\vdots\\0\\
0\end{array}
\right),
\end{align*}
and 
\begin{align*}
\hat{\tilde{M}}_j=\left(\begin{array}{cccc}
d_j&0&\cdots&0\\
0&d_{j+1}&\cdots&0\\
\vdots&\vdots&\ddots&\vdots\\
0&0&\cdots&d_{j+\theta-1}\end{array}
\right),\quad \hat{M}_j=\left(\begin{array}{c}
\hat{\tilde{M}}_j\\0\end{array}
\right),\quad \tilde{M}_j=\left(
\hat{\tilde{M}}_j,\,\, 0
\right),\quad M_j=\left(\begin{array}{cc}
\hat{\tilde{M}}_j&0\\
0&d_{j+\theta}
\end{array}
\right).
\end{align*}
The action of a permutation matrix gives the result.

\end{proof}
Therefore, we know that $\{H_n\}_{n\in\mathbb{N}}$ could be divided into $\theta+1$ different families of Wronski quasi-determinant. Let's denote
\begin{align*}
H_{n(\theta+1)+\ell}:=\tilde{H}_n^{(\ell)},\quad \ell=0,1,\cdots,\theta,\,n\in\mathbb{N}
\end{align*}
and if moments $\{d_i,\,i=0,1,\cdots\}$ satisfy evolutions $\p_t d_i=d_{i+\theta}$, we know that for each $\ell=0,\cdots,\theta$, $\{\tilde{H}_n^{(\ell)}\}_{n\in\mathbb{N}}$ are solutions of the non-commutative Blaszak-Marciniak and Kuperschmidt lattices based on Proposition \ref{propbt}.

Our next step is to find a connecting equation for these solutions. For this purpose, we consider how moment reduction influences the matrix-valued bi-orthogonal polynomials. 
\begin{theorem}
Under the moment reduction condition \eqref{mred}, we have
\begin{align*}
P_{n(\theta+1)+\ell}(x)=x^\ell\tilde{P}^{(\ell)}_n(x^{\theta+1}),\quad Q_{n(\theta+1)+\ell}=x^\ell\tilde{Q}_n^{(\ell)}(x^{\theta+1}),\quad \ell=0,1,\cdots,\theta,
\end{align*}
where
\begin{align*}
&\tilde{P}_n^{(\ell)}(x)=\left|\begin{array}{ccccc}
d_\ell&d_{\ell+\theta}&\cdots&d_{\ell+(n-1)\theta}&\mathbb{I}_p\\
d_{\ell+1}&d_{\ell+1+\theta}&\cdots&d_{\ell+1+(n-1)\theta}&x\mathbb{I}_p\\
\vdots&\vdots&&\vdots&\vdots\\
d_{\ell+n}&d_{\ell+n+\theta}&\cdots&d_{\ell+n+(n-1)\theta}&\boxed{x^n\mathbb{I}_p}
\end{array}
\right|,\\
&\tilde{Q}_n^{(\ell)}(x)=\left|\begin{array}{cccc}
d_\ell&d_{\ell+\theta}&\cdots&d_{\ell+n\theta}\\
\vdots&\vdots&&\vdots\\
d_{\ell+n-1}&d_{\ell+n-1+\theta}&\cdots&d_{\ell+n-1+n\theta}\\
\mathbb{I}_p&x\mathbb{I}_p&\cdots&\boxed{x^n\mathbb{I}_p}
\end{array}
\right|.
\end{align*}
Moreover,  the graded orthogonality reads
\begin{align}\label{or2}
\langle P_{n(\theta+1)+\ell}(x),Q_{m(\theta+1)+k}(x)\rangle_\theta=\tilde{H}_n^{(\ell)}\delta_{n,m}\delta_{\ell,k}.
\end{align}
\end{theorem}
The proof of this theorem is based on the same idea as in Prop. \ref{prop5.2}, by making use of block matrices and permutations. Moreover, this theorem implies that there is a gradation of the module
\begin{align*}
\mathbb{R}^{p\times p}[x]= \mathbb{R}_0^{p\times p}[x]\oplus \mathbb{R}_1^{p\times p}[x]\oplus\cdots\oplus \mathbb{R}_\theta^{p\times p}[x],
\end{align*}
where
\begin{align*}
\mathbb{R}_\ell^{p\times p}[x]=\text{span}\{x^\ell\mathbb{I}_p,x^{\ell+\theta+1}\mathbb{I}_p,\cdots,x^{\ell+n(\theta+1)}\mathbb{I}_p,\cdots\},\quad \ell=0,1,\cdots,\theta.
\end{align*}
Therefore, we have the following important observation
\begin{align}\label{span}
P_{n(\theta+1)+\ell}(x)\in\mathbb{R}^{p\times p}_\ell[x],
\end{align}
from which the terms of recurrence relations would be dramatically decreased. We have the following proposition.
\begin{proposition}
Monic bi-orthogonal polynomials $\{P_n(x)\}_{n\in\mathbb{N}}$ satisfying the orthogonal relation \eqref{or2},  have recurrence relations
\begin{align}\label{ss1}
\begin{aligned}
x^\theta P_{n(\theta+1)}(x)&=P_{n(\theta+1)+\theta}(x)+\xi_{n,0}P_{(n-1)(\theta+1)+\theta}(x),\\
x^\theta P_{n(\theta+1)+\ell}(x)&=P_{(n+1)(\theta+1)+\ell-1}(x)+\xi_{n,\ell}P_{n(\theta+1)+\ell-1}(x),\quad \ell=1,\cdots,\theta,
\end{aligned}
\end{align}
where 
\begin{align*}
\xi_{n,\ell}=\left\{
\begin{array}{ll}
\tth_n^{(0)}\left(
\tth_{n-1}^{(\theta)}
\right)^{-1},& \ell=0,\\
\tth_n^{(\ell)}\left(
\tth_n^{(\ell-1)}
\right)^{-1},&\ell=1,\cdots,\theta
\end{array}
\right.
\end{align*}
The $\{Q_n(x)\}_{n\in\mathbb{N}}$,  satisfy
\begin{align}\label{ss2}
\begin{aligned}
xQ_{n(\theta+1)+\ell}&=Q_{n(\theta+1)+\ell+1}(x)+\eta_{n,\ell}Q_{(n-1)(\theta+1)+\ell+1},\quad \ell=0,\cdots,\theta-1,\\
xQ_{n(\theta+1)+\theta}&=Q_{(n+1)(\theta+1)}(x)+\eta_{n,\theta}Q_{n(\theta+1)},
\end{aligned}
\end{align}
where 
\begin{align*}
\eta_{n,\ell}^\top=\left\{
\begin{array}{ll}
\left(\tth_{n-1}^{(\ell+1)}\right)^{-1}\tth_n^{(\ell)},& \ell=0,\cdots,\theta-1,\\
\left(\tth_n^{(0)}\right)^{-1}\tth_n^{(\theta)},&\ell=\theta.
\end{array}
\right.
\end{align*}
\end{proposition}
\begin{proof}
When $x^\theta$ acts on $P_{n(\theta+1)+\ell}(x)$, there are two different cases based on whether $\ell+\theta\leq \theta$ or $\ell+\theta>\theta$ due to the gradation. Therefore, we should discuss $\ell=0$ and $\ell=1,\cdots,\theta$ seperately. When $\ell=0$, then
\begin{align*}
x^\theta P_{n(\theta+1)}(x)=P_{n(\theta+1)+\theta}(x)+\sum_{k=0}^{n-1}\xi_{n,0}^{(k)}P_{k(\theta+1)+\theta}(x),
\end{align*}
where
\begin{align*}
\xi_{n,0}^{(k)}=\langle x^\theta P_{n(\theta+1)}(x),Q_{k(\theta+1)+\theta}(x)\rangle_\theta \left(\tth_k^{(\theta)}\right)^{-1}=\langle P_{n(\theta+1)}(x),xQ_{k(\theta+1)+\theta}(x)\rangle_\theta \left(\tth_k^{(\theta)}\right)^{-1},
\end{align*}
and the last step here is due to the quasi-symmetry \eqref{qs}. Therefore, according to the graded orthogonality  \eqref{or2}, only when $k=n-1$ is the bilinear form  nonzero and we get 
\begin{align*}
\xi_{n,0}^{(k)}=\left\{\begin{array}{ll}
0,&k=0,\cdots,n-2,\\
\tth_{n}^{(0)}\left(\tth_{n-1}^{(\theta)}\right)^{-1},&k=n-1.\end{array}\right.
\end{align*}
When $\ell=1,\cdots,\theta$, by similar discussion, we have
\begin{align*}
x^\theta P_{n(\theta+1)+\ell}(x)=P_{(n+1)(\theta+1)+\ell-1}(x)+\sum_{k=0}^n \xi_{n,\ell}^{(k)}P_{k(\theta+1)+\ell-1}(x),
\end{align*}
and 
\begin{align*}
\xi_{n,\ell}^{(k)}=\left\{\begin{array}{ll}
0&k=0,\cdots,n-1,\\
\tth_n^{(\ell)}\left(
\tth_n^{(\ell-1)}
\right)^{-1}&k=n.
\end{array}
\right.
\end{align*}
Therefore, we can ignore the upper index and the proof is complete for $\{P_n(x)\}_{n\in\mathbb{N}}$. The proof for $\{Q_n(x)\}_{n\in\mathbb{N}}$ is similar.
\end{proof}
\begin{remark}
The corresponding spectral problems \eqref{ss1} and \eqref{ss2} can  be characterized by
\begin{subequations}
\begin{align}
&x^\theta \Phi=(\Lambda^\theta+\mathcal{A}\Lambda^{-1})\Phi:=\mathcal{L}\Phi, &\mathcal{A}=\text{diag}(\xi_{0,0},\cdots,\xi_{0,\theta},\xi_{1,0},\cdots,\xi_{1,\theta},\cdots),\label{sp1}\\
&x\Psi=(\Lambda+\mathcal{B}\Lambda^{-\theta})\Psi:=\mathcal{M}\Psi,&\mathcal{B}=\text{diag}(\eta_{0,0},\cdots,\eta_{0,\theta},\eta_{1,0},\cdots,\eta_{1,\theta},\cdots).\label{sp2}
\end{align}
\end{subequations}
\end{remark}
Time evolutions are then considered. Considering the time-dependent weight function \eqref{td}, we have the following proposition.
\begin{proposition}
Only $t_{c(\theta+1)}$-flows ($c=1,2,\cdots$) are compatible with the moment reduction condition \eqref{mred}. In other words, if $k \,\text{mod} \,(\theta+1)\ne0$,
\begin{align*}
\p_{t_k} P_n(x;\mt)=\p_{t_k}Q_n(x;\mt)=0,
\end{align*}
i.e. the wave function doesn't evolve with respect to $t_k$-flows if $k \,\text{mod} \,(\theta+1)\ne0$.
\end{proposition}
\begin{proof}
From the expression for the  time-dependent weight function \eqref{td}, we know that $\p_{t_k}m_{i}=m_{i+k}$. With the moment reduction condition, it is known that if $i=0$ mod ($\theta+1$), then $i+k =k$ mod $(\theta+1)$. Therefore, only when $k=0$ mod ($\theta+1$) do we have $\p_{t_k}d_i\ne0$, and the corresponding derivative of the coefficient doesn't vanish.
\end{proof}
Therefore, if we take the $t_{c(\theta+1)}$-derivative of the orthogonal relation $\left(\langle \Phi,\Psi\rangle_\theta\right)=\mathcal{H}$, we have
\begin{align*}
\left(
\langle \p_{t_{c(\theta+1)}}\Phi,\Psi\rangle_\theta
\right)+\left(
\langle \Phi,\p_{t_{c(\theta+1)}}\Psi\rangle_\theta
\right)+\left(\langle
x^{c(\theta+1)}\Phi,\Psi\rangle_\theta
\right)=\p_{t_{c(\theta+1)}}\mathcal{H}.
\end{align*}
Moreover, if $c$ is a multiplier of $\theta$, then we get
\begin{subequations}
\begin{align}
\p_{t_{k\theta(\theta+1)}}\Phi&=-(\mathcal{L}^{k(\theta+1)})_{<0}\Phi,\label{time1}\\
\p_{t_{k\theta(\theta+1)}}\Psi&=-(\mathcal{M}^{k(\theta+1)})_{<0}\Psi,\label{time2}
\end{align}
\end{subequations}
by following the proof of Proposition \ref{prop4}.
The first non-trivial flow is $t_{\theta(\theta+1)}$ flow, which could be read by the compatibility conditions of \eqref{sp1}-\eqref{time1}, and \eqref{sp2}-\eqref{time2}. Thus we get
\begin{subequations}
\begin{align}
\p_{t_{\theta(\theta+1)}}\mathcal{L}&=\left[\mathcal{L},(\mathcal{L}^{\theta+1})_{<0}\right],\label{inb1}\\
\p_{t_{\theta(\theta+1)}}\mathcal{M}&=\left[\mathcal{M},(\mathcal{M}^{\theta+1})_{<0}\right].\label{inb2}
\end{align}
\end{subequations}

By realizing that
\begin{align*}
\left((\Lambda^\theta+\mathcal{A}\Lambda^{-1})^{\theta+1}\right)
_{<0}=\left((\mathcal{A}\Lambda^{-1})^{\theta+1}\right)_{<0},
\end{align*}
we get the equation expressed in terms of elements in $\mathcal{A}$ and
\begin{align*}
\p_{t_{\theta(\theta+1)}}(\mathcal{A}\Lambda^{-1})=\left[\Lambda^\theta,(\mathcal{A}\Lambda^{-1})^{\theta+1}
\right].
\end{align*}
Moreover, if we denote $\xi_{i,j}=\gamma_{i(\theta+1)+j}$, then equation \eqref{inb1} could be explicitly written as
\begin{align}\label{bog1}
\p_{t_{\theta(\theta+1)}}\gamma_n=
\gamma_{n+\theta}\cdots\gamma_n-\gamma_n\cdots\gamma_{n-\theta}.
\end{align}
On the other hand, from equations \eqref{inb2}, we obtain the equation in terms of $\mathcal{B}$ and 
\begin{align*}
\p_{t_{\theta(\theta+1)}}\mathcal{B}\Lambda^{-\theta}=\sum_{{i+j=\theta-1 \atop i,j\geq 0}}\Lambda^{i+1}(\mathcal{B}\Lambda^{-\theta})\Lambda^j (\mathcal{B}\Lambda^{-\theta})-\sum_{{i+j=\theta-1 \atop i,j\geq 0}}(\mathcal{B}\Lambda^{-\theta})\Lambda^i (\mathcal{B}\Lambda^{-\theta})\Lambda^{j+1}.
\end{align*}
Explicitly, one has
\begin{align}\label{bog2}
\p_{t_{\theta(\theta+1)}}\zeta_n=\left(
\sum_{i=1}^\theta \zeta_{n+i}
\right)\zeta_n-\zeta_n\left(
\sum_{i=1}^\theta \zeta_{n-i}
\right)
\end{align}
where $\eta_{i,j}=\zeta_{i(\theta+1)+j}$.

In the literature, the commutative versions of equations \eqref{bog1} and \eqref{bog2} are called the Itoh-Narita-Bogoyavleskii (INB) lattices \cite{bogoyavlensky91, casati21, itoh75, narita82}, as additive and multiplicative extensions of the  Lotka-Volterra lattice. In \cite{casati21}, the Hamiltonian structure of non-commutative INB lattice was studied recently. 

\subsection{The B\"acklund transformation for Blaszak-Marciniak three-field equations and its quasi-determinant solutions}\label{sec4.2}
In this section, let us consider the $\theta=2$ case, which gives a B\"acklund transformation for the three-field equations \eqref{nc-bm-2-1}-\eqref{nc-bm-2-3}. It reads from \eqref{bog1} that\footnote{Here, we use $t$ instead of $t_6$ to make the writing simpler.}
\begin{align}\label{bo1}
\p_{t}\gamma_n=\gamma_{n+2}\gamma_{n+1}\gamma_n-\gamma_n\gamma_{n-1}\gamma_{n-2},
\end{align}
and we have the following proposition.
\begin{proposition}
The non-commutative INB lattice \eqref{bo1} has the solutions
\begin{align*}
\gamma_{3n}=\tth_n^{(0)}\left(\tth_{n-1}^{(2)}\right)^{-1},\quad \gamma_{3n+1}=\tth_n^{(1)}\left(\tth_n^{(0)}\right)^{-1},\quad \gamma_{3n+2}=\tth_n^{(2)}\left(
\tth_n^{(1)}
\right)^{-1},
\end{align*}
whose moments $\{d_i\}_{i\in\mathbb{N}}$ satisfy time evolutions $\p_{t}d_i=d_{i+2}$.
\end{proposition}

If we take solutions into the equation \eqref{bo1}, it is necessary to verify that
\begin{subequations}
\begin{align}
\p_t\left(
\tth_n^{(0)}\left(\tth_{n-1}^{(2)}\right)^{-1}
\right)&=\tth_n^{(2)}\left(
\tth_{n-1}^{(2)}
\right)^{-1}-\tth_n^{(0)}\left(\tth_{n-1}^{(0)}\right)^{-1},\label{sub5-1}\\
\p_t\left(
\tth_n^{(1)}\left(\tth_{n}^{(0)}\right)^{-1}
\right)&=\tth_{n+1}^{(0)}\left(
\tth_{n}^{(0)}
\right)^{-1}-\tth_{n}^{(1)}\left(\tth_{n-1}^{(1)}\right)^{-1},\label{sub5-2}\\
\p_t\left(
\tth_n^{(2)}\left(\tth_{n}^{(1)}\right)^{-1}
\right)&=\tth_{n+1}^{(1)}\left(
\tth_{n}^{(1)}
\right)^{-1}-\tth_n^{(2)}\left(\tth_{n-1}^{(2)}\right)^{-1}.\label{sub5-3}
\end{align}
\end{subequations}
By simplifying \eqref{sub5-1}, one gets
\begin{align}\label{altt1}
\left(
\tth_n^{(0)}
\right)^{-1}\p_t \tth_n^{(0)}-\left(\tth_{n-1}^{(2)}
\right)^{-1}\p_t \tth_{n-1}^{(2)}=\left(
\tth_n^{(0)}
\right)^{-1}\tth_n^{(2)}-\left(
\tth_{n-1}^{(0)}
\right)^{-1}H_{n-1}^{(2)}.
\end{align}
According to equations \eqref{pt-H} and \eqref{eta-z}, we know that if we introduce the notation
\begin{align*}
\teta_{n+1,n}^{(\ell)}=\left|\begin{array}{ccccc}
d_\ell&d_{\ell+2}&\cdots&d_{\ell+2n}&d_{\ell+2n+2}\\
d_{\ell+1}&d_{\ell+3}&\cdots&d_{\ell+2n+1}&d_{\ell+2n+3}\\
\vdots&\vdots&&\vdots&\vdots\\
d_{\ell+n}&d_{\ell+n+2}&\cdots&d_{\ell+3n}&d_{\ell+3n+2}\\
0&0&\cdots&\mathbb{I}_p&\boxed{0}\end{array}
\right|,
\end{align*}
then \eqref{altt1} is equal to
\begin{align*}
\teta_{n,n-1}^{(2)}-\teta_{n+1,n}^{(0)}-\left(
\teta_{n-1,n-2}^{(2)}-\teta_{n,n-1}^{(0)}
\right)=\left(
\tth_n^{(0)}
\right)^{-1}\tth_n^{(2)}-\left(
\tth_{n-1}^{(0)}
\right)^{-1}\tth_{n-1}^{(2)}.
\end{align*}
Moreover, from the next proposition, we know that \eqref{sub5-1} is valid.
\begin{proposition}
It holds that
\begin{align*}
\teta_{n+1,n}^{(0)}=\teta_{n,n-1}^{(2)}-\left(
\tth_n^{(0)}
\right)^{-1}\tth_n^{(2)}.
\end{align*}
\end{proposition}
\begin{proof}
According to the identity \eqref{ncj1} to $(n,n+1)$-rows and $(1,n+1)$-columns, we have
\begin{align*}
\teta_{n+1,n}^{(0)}&=\left|\begin{array}{cccc}
d_2&\cdots&d_{2n}&d_{2n+2}\\
d_3&\cdots&d_{2n+1}&d_{2n+3}\\
\vdots&&\vdots&\vdots\\
d_{n+2}&\cdots&d_{3n}&d_{3n+2}\\
0&\cdots&\mathbb{I}_p&\boxed{0}
\end{array}
\right|\\&-\left|\begin{array}{cccc}
d_0&d_2&\cdots&d_{2n}\\
d_1&d_3&\cdots&d_{2n+1}\\
\vdots&\vdots&&\vdots\\
d_{n-1}&d_{n+1}&\cdots&d_{3n-1}\\
\boxed{0}&0&\cdots&\mathbb{I}_p
\end{array}\right|\left|\begin{array}{cccc}
d_0&d_2&\cdots&d_{2n}\\
d_1&d_3&\cdots&d_{2n+1}\\
\vdots&\vdots&&\vdots\\
d_{n-1}&d_{n+1}&\cdots&d_{3n-1}\\
\boxed{d_n}&d_{n+2}&\cdots&d_{3n}
\end{array}
\right|^{-1}\left|\begin{array}{cccc}
d_2&\cdots&d_{2n}&d_{2n+2}\\
d_3&\cdots&d_{2n+1}&d_{2n+3}\\
\vdots&&\vdots&\vdots\\
d_{n+1}&\cdots&d_{3n-1}&d_{3n+1}\\
d_{n+2}&\cdots&d_{3n}&\boxed{d_{3n+2}}
\end{array}
\right|.
\end{align*}
Moreover, by applying the homological relation \eqref{hm1}, this proof is complete.
\end{proof}
We note that  proving  equations \eqref{sub5-2} and \eqref{sub5-3}, is equivalent to verifying
\begin{align*}
\left(
H_n^{(\ell+1)}
\right)^{-1}\p_t H_n^{(\ell+1)}-\left(
H_n^{(\ell)}
\right)^{-1}\p_t H_n^{(\ell)}=\left(
H_n^{(\ell+1)}
\right)^{-1}H_{n+1}^{(\ell)}-\left(
H_{n-1}^{(\ell+1)}
\right)^{-1}H_n^{(\ell)},\quad \ell=0,1.
\end{align*}
By using \eqref{pt-H} and \eqref{eta-z},
it is known that we only need to prove the following proposition.
\begin{proposition}
It holds that
\begin{align*}
\teta_{n+1,n}^{(\ell)}-\teta_{n+1,n}^{(\ell+1)}=\left(
\tth_{n}^{(\ell+1)}
\right)^{-1}\tth_{n+1}^{(\ell)},\quad \ell=0,1.
\end{align*}
\end{proposition}

\begin{proof}
This proof is based on the observation that
\begin{align*}
\left|\begin{array}{cccccc}
d_\ell&d_{\ell+2}&\cdots&d_{\ell+2n}&d_{\ell+2n+2}&0\\
d_{\ell+1}&d_{\ell+3}&\cdots&d_{\ell+2n+1}&d_{\ell+2n+3}&0\\
\vdots&\vdots&&\vdots&\vdots&\vdots\\
d_{\ell+n}&d_{\ell+n+2}&\cdots&d_{\ell+3n}&d_{\ell+3n+2}&0\\
d_{\ell+n+1}&d_{\ell+n+3}&\cdots&d_{\ell+3n+1}&d_{\ell+3n+3}&\mathbb{I}_p\\
0&0&\cdots&\mathbb{I}_p&0&\boxed{0}
\end{array}
\right|&=-\teta_{n+1,n}^{(\ell)}\left(
\tth_{n+1}^{(\ell)}
\right)^{-1}\\
&=-\left(
\tth_n^{(\ell+1)}
\right)^{-1}-\teta_{n+1,n}^{(\ell+1)}\left(
\tth_{n+1}^{(\ell)}
\right)^{-1},
\end{align*}
where two different non-commutative Jacobi identities are applied to the left-hand side quasi-determinants. The first equality is the application of non-commutative Jacobi identity \eqref{ncj1} to $(n+1,n+2)$-rows and $(n+1,n+2)$-columns while the second equality is obtained by using \eqref{ncj1} to $(1,n+2)$-rows and $(n+1,n+2)$-columns.
\end{proof}

\subsection{B\"acklund transformation for bi-graded Toda lattices and fractional Volterra hierarchy}\label{sec4.3}
This part is devoted to finding the B\"acklund transformation for the case where $\theta=b/a\in\mathbb{Q}_+$ with $a, b\in\mathbb{Z}_+$. Consider the moment reduction condition
\begin{align}\label{mred2}
\langle x^i\mathbb{I}_p,x^j\mathbb{I}_p\rangle_{\frac{b}{a}}=m_{i+j\frac{b}{a}}:=\left\{\begin{array}{ll}
d_{\frac{ia+jb}{a(a+b)}},& ia+jb \mod a+b =0,\\
0,& ia+jb \mod a+b\ne 0.
\end{array}
\right.
\end{align}
By a similar manner to the proof of Prop. \ref{prop5.2}, we claim the following proposition.
\begin{proposition}
Under the moment reduction condition \eqref{mred2}, we have
\begin{align*}
H_{n(a+b)+\ell}=\left|\begin{array}{cccc}
d_{\frac{\ell}{a}}&d_{\frac{\ell}{a}+\theta}&\cdots&d_{\frac{\ell}{a}+n\theta}\\
d_{\frac{\ell}{a}+1}&d_{\frac{\ell}{a}+1+\theta}&\cdots&d_{\frac{\ell}{a}+1+n\theta}\\
\vdots&\vdots&&\vdots\\
d_{\frac{\ell}{a}+n}&d_{\frac{\ell}{a}+n+\theta}&\cdots&\boxed{d_{\frac{\ell}{a}+n+n\theta}}
\end{array}
\right|:=\tilde{H}_n^{(\ell)},\quad \ell=0,1,\cdots,a+b-1,
\end{align*}
and for the polynomials, we have
\begin{align*}
P_{n(a+b)+\ell}(x)=x^\ell \tilde{P}^{(\ell)}_n(x^{a+b}),\quad Q_{n(a+b)+\ell}(x)=x^\ell \tilde{Q}_n^{(\ell)}(x^{a+b}),\quad \ell=0,1,\cdots,a+b-1,
\end{align*}
where
\begin{align*}
\tilde{P}_n^{(\ell)}(x)=\left|\begin{array}{ccccc}
d_{\frac{\ell}{a}}&d_{\frac{\ell}{a}+\theta}&\cdots&d_{\frac{\ell}{a}+(n-1)\theta}&\mathbb{I}_p\\
d_{\frac{\ell}{a}+1}&d_{\frac{\ell}{a}+1+\theta}&\cdots&d_{\frac{\ell}{a}+1+(n-1)\theta}&x^{a+b}\mathbb{I}_p\\
\vdots&\vdots&&\vdots&\vdots\\
d_{\frac{\ell}{a}+n}&d_{\frac{\ell}{a}+n+\theta}&\cdots&d_{\frac{\ell}{a}+n+(n-1)\theta}&\boxed{x^{n(a+b)}\mathbb{I}_p}
\end{array}
\right|,
\end{align*}
and
\begin{align*}
\tilde{Q}_n^{(\ell)}(x)=\left|\begin{array}{ccccc}
d_{\frac{\ell}{a}}&d_{\frac{\ell}{a}+\theta}&\cdots&d_{\frac{\ell}{a}+(n-1)\theta}&d_{\frac{\ell}{a}+n\theta}\\
d_{\frac{\ell}{a}+1}&d_{\frac{\ell}{a}+1+\theta}&\cdots&d_{\frac{\ell}{a}+1+(n-1)\theta}&d_{\frac{\ell}{a}+1+n\theta}\\
\vdots&\vdots&&\vdots&\vdots\\
d_{\frac{\ell}{a}+n-1}&d_{\frac{\ell}{a}+n-1+\theta}&\cdots&d_{\frac{\ell}{a}+n-1+(n-1)\theta}&d_{\frac{\ell}{a}+n-1+n\theta}\\
\mathbb{I}_p&x^{a+b}\mathbb{I}_p&\cdots&x^{(n-1)(a+b)}\mathbb{I}_p&\boxed{x^{n(a+b)}\mathbb{I}_p}
\end{array}
\right|.
\end{align*}
\end{proposition}
Moreover, orthogonal relations under this moment reduction condition can be written as 
\begin{align}\label{oc5}
\langle P_{n(a+b)+\ell}(x),Q_{m(a+b)+k}(x)\rangle_\theta=\tilde{H}_n^{(\ell)}\delta_{n,m}\delta_{\ell,k}.
\end{align}
In analogy to \eqref{span}, we know that there is a graded polynomial space 
\begin{align*}
\mathbb{R}^{p\times p}[x]= \mathbb{R}_0^{p\times p}[x]\oplus \mathbb{R}_1^{p\times p}[x]\oplus\cdots\oplus \mathbb{R}_{a+b-1}^{p\times p}[x],
\end{align*}
and
\begin{align*}
P_{n(a+b)+\ell}(x)\in \text{span}\{x^\ell\mathbb{I}_p,x^{\ell+a+b}\mathbb{I}_p,\cdots,x^{\ell+n(a+b)}\mathbb{I}_p\}.
\end{align*}
According to the quasi-symmetry property \eqref{qscon}, we have the following recurrence relations.
\begin{proposition}
For monic bi-orthogonal polynomials $\{P_n(x)\}_{n\in\mathbb{N}}$ satisfying the orthogonal relation \eqref{oc5}, we have recurrence relations
\begin{align*}
x^bP_{n(a+b)+\ell}(x)=P_{n(a+b)+b+\ell}(x)+\xi_{n,\ell}P_{(n-1)(a+b)+b+\ell}(x),\quad \ell=0,\cdots,a-1
\end{align*}
with $\xi_{n,\ell}=\tilde{H}_n^{(\ell)}\left(
\tilde{H}_{n-1}^{(b+\ell)}
\right)^{-1}$, and 
\begin{align*}
x^bP_{n(a+b)+\ell}(x)=P_{(n+1)(a+b)+\ell-a}(x)+\xi_{n,\ell}P_{n(a+b)+\ell-a}(x),\quad \ell=a,\cdots,a+b
\end{align*}
with $\xi_{n,\ell}=\tilde{H}_n^{(\ell)}\left(
\tilde{H}_n^{(\ell-a)}
\right)^{-1}$. Similarly the $\{Q_n(x)\}_{n\in\mathbb{N}}$,   satisfy the recurrence relations
\begin{align*}
x^aQ_{n(a+b)+\ell}(x)=Q_{n(a+b)+a+\ell}(x)+\eta_{n,\ell}Q_{(n-1)(a+b)+a+\ell}(x),\quad \ell=0,\cdots,b-1
\end{align*}
with $\eta_{n,\ell}^\top=\left(
\tilde{H}_{n-1}^{(a+\ell)}
\right)^{-1}\tilde{H}_n^{(\ell)}$, and
\begin{align*}
x^aQ_{n(a+b)+\ell}(x)=Q_{(n+1)(a+b)+\ell-b}(x)+\eta_{n,\ell}Q_{n(a+b)+\ell-b}(x),\quad \ell=b,\cdots,a+b
\end{align*}
with
$\eta_{n,\ell}^\top=\left(
\tilde{H}_n^{(\ell-b)}
\right)^{-1}\tilde{H}_n^{(\ell)}$.
\end{proposition}
If we denote $\Phi=(P_0(x),P_1(x),\cdots)$ and $\Psi=(Q_0(x),Q_1(x),\cdots)$, then the corresponding recurrence relations can be written into spectral problem form
\begin{align*}
x^b\Phi&=(\Lambda^b+\mathcal{A}\Lambda^{-a})\Phi, \quad \mathcal{A}=\text{diag}(\xi_{0,0},\cdots,\xi_{0,a+b},\xi_{1,0},\cdots,\xi_{1,a+b},\cdots),\\
x^a\Psi&=(\Lambda^a+\mathcal{B}\Lambda^{-b})\Psi,\quad \mathcal{B}=\text{diag}(\eta_{0,0},\cdots,\eta_{0,a+b},\eta_{1,0},\cdots,\eta_{1,a+b},\cdots).
\end{align*}
Such a spectral problem corresponds to the so-called fractional Volterra hierarchy, which was studied in the cubic Hodge integrals and integrable systems in \cite{liu18}.

In particular, if we consider $t_{kb(a+b)}$-flows, then we obtain the corresponding integrable lattices
\begin{align*}
\p_{t_{kb(a+b)}}\mathcal{L}=\left[\mathcal{L},\left(
\mathcal{L}^{k(a+b)}
\right)_{<0}\right],\quad 
\p_{t_{kb(a+b)}}\mathcal{M}=\left[\mathcal{M},\left(
\mathcal{M}^{k(a+b)}
\right)_{<0}\right],
\end{align*}
where $\mathcal{L}=\Lambda^b+\mathcal{A}\Lambda^{-a}$ and $\mathcal{M}=\Lambda^a+\mathcal{B}\Lambda^{-b}$.

\section*{Acknowledgement}
This work is partially funded by grants (NSFC12101432, NSFC12175155). SHL would like to thank Dr. Jesper Ipsen for fruitful discussions on the Muttalib-Borodin model and fractional powers of bi-orthogonal functions. SHL would also like to thank Prof. Di Yang for sharing his ideas on the bigraded Toda hierarchy and applications in Frobenius manifolds.

\appendix\label{appendixa}
\section*{Appendix A}
\renewcommand{\thesection}{A} 
\setcounter{equation}{0}
\setcounter{theorem}{0}
\subsection{Basic quasi-determinant identities}\label{AppendixA1}

Quasi-determinants were first introduced by Gelfand and Retakh in the early 1990s for a matrix with non-commutative entries \cite{gelfand91}. Since this paper deals with matrix-valued orthogonal polynomials and non-commutative integrable systems by using quasi-determinants, we give a brief introduction to quasi-determinants as well as the basic properties we used in this paper. For a detailed reference, please see \cite{gelfand91,gelfand05,krob95}.

\begin{definition}\label{qd-def}
Let $A$ be an $n\times n$ matrix over a ring $\mathcal{R}$.
For $i, j=1, 2, \dots, n$, let $r_i^j$ be the $i$-th row of $A$ without the $j$-th entry, $c_j^i$ be the $j$-th column without the $i$-th entry, and $A^{i,j}$ be the submatrix of $A$ without the $i$-th row and $j$-th column of $A$.
Assume that $A^{i,j}$ is invertible.
Then there are $n^2$ quasi-determinants of $A$, denoted as $|A|_{i,j}$ for $1\leq i, j\leq n$, as follows
\begin{align*}
|A|_{i,j}=a_{i,j}-r_i^j\left(A^{i,j}\right)^{-1}c_j^i,
\end{align*}
where $a_{i,j}$ is the $(i,j)$-th entry of $A$. For convenience, in this paper, we denote 
\begin{align*}
|A|_{i,j}=\left|\begin{array}{cc}
A^{i,j}&c_j^i\\
r_i^j&\boxed{a_{ij}}
\end{array}
\right|.
\end{align*}
\end{definition}
Since the concept of quasi-determinant was proposed, it has been used in different branches of mathematics, such as representation theory, combinatorics, non-commutative geometry and so on. For self-consistency, we list below, several basic properties of quasi-determinants which we have used in this article.

\begin{enumerate}[0]
\item[$\bullet$]
Quasi-determinants can be used to solve linear systems with non-commutative coefficients.

\begin{proposition}(\cite[Thm. 1.6.1]{gelfand05})\label{p-ls}
Let $A=(a_{i,j})$ be an $n\times n$ matrix over a ring $\mathcal{R}$.
Assume that all the quasi-determinants $|A|_{i,j}$  are defined and invertible.
Then
\begin{align*}
\left\{\begin{array}{c}
a_{1,1}x_1+\cdots+a_{1,n}x_n=\xi_1\\
\vdots\\
a_{n,1}x_1+\cdots+a_{n,n}x_n=\xi_n\end{array}
\right.
\end{align*}
has a solution $x_i\in\mathcal{R}$ if and only if
\begin{align*}
x_i=\sum_{j=1}^n|A|_{j,i}^{-1}\xi_j.
\end{align*}
\end{proposition}
\end{enumerate}

\begin{enumerate}[0]
\item[$\bullet$]  
Invariance of quasi-determinants under elementary row or column operations.

\begin{proposition}(\cite[Prop 2.2]{krob95})\label{p-inva}
A permutation of the rows or columns of a quasi-determinant does not change its value.
\end{proposition}
For example, 
\begin{align*}
\left|\begin{array}{ccc}
A&B&C\\
D&f&g\\
E&h&\boxed{i}\end{array}
\right|
=\left|\begin{array}{ccc}
B&A&C\\
f&D&g\\
h&E&\boxed{i}
\end{array}
\right|
=\left|\begin{array}{ccc}
A&B&C\\
E&h&\boxed{i}\\
D&f&g
\end{array}
\right|.
\end{align*}
Moreover, this proposition could be written in a general form (\cite[Eq. 8]{gilson07})
\begin{align*}
\left|
\left(\begin{array}{cc}
F&0\\
E&h
\end{array}
\right)
\left(\begin{array}{cc}
A&B\\
D&f
\end{array}
\right)
\right|_{n,n}
=h\left|\begin{array}{cc}
A&B\\
D&\boxed{f}
\end{array}
\right|.
\end{align*}
\end{enumerate}

\begin{enumerate}[0]
\item[$\bullet$]  
Equivalent conditions for a zero quasi-determinant.
\begin{proposition}(\cite[Prop. 1.4.6]{gelfand05})\label{p-equi}
The following statements are equivalent if the quasi-determinant $|A|_{ij}$ is defined.\\
({\romannumeral1})~~$|A|_{ij}$=0;\\
({\romannumeral2})~~The $i$-th row of the matrix $A$ is a left linear combination of the other rows of $A$.\\
({\romannumeral3})~~The $j$-th column of the matrix $A$ is a right linear combination of the other columns of $A$.
\end{proposition}
\end{enumerate}

\begin{enumerate}[0]
\item[$\bullet$]  
Non-commutative Jacobi identity.
\noindent
There are several identities for quasi-determinants, as an analogy of Sylvester identity for determinants. For a general Sylvester's identity for quasi-determinants, please refer to \cite{gelfand91}. In this article, we mainly make use of the simplest one, which is called the non-commutative Jacobi identity \cite{ gilson07}
\begin{align}\label{ncj1}
\left|\begin{array}{ccc}
A&B&C\\
D&f&g\\
E&h&\boxed{i}
\end{array}
\right|
=\left|
\begin{array}{cc}
A&C\\
E&\boxed{i}
\end{array}
\right|-\left|\begin{array}{cc}
A&B\\
E&\boxed{h}
\end{array}
\right|\left|\begin{array}{cc}
A&B\\
D&\boxed{f}
\end{array}
\right|^{-1}\left|\begin{array}{cc}
A&C\\
D&\boxed{g}
\end{array}
\right|.
\end{align}
According to Prop. \ref{p-inva}, it has the following alternative form
\begin{align}\label{ncj2}
\left|\begin{array}{ccc}
B&A&C\\
f&D&g\\
h&E&\boxed{i}
\end{array}
\right|=\left|\begin{array}{cc}
A&C\\
E&\boxed{i}
\end{array}
\right|-\left|\begin{array}{cc}
B&A\\
\boxed{h}&E
\end{array}
\right|\left|\begin{array}{cc}
B&A\\
\boxed{f}&D
\end{array}
\right|^{-1}\left|\begin{array}{cc}
A&C\\
D&\boxed{g}
\end{array}
\right|.
\end{align}
\end{enumerate}

\begin{enumerate}[0]
\item[$\bullet$]  
Homological relations in terms of quasi-Pl\"ucker coordinates.

\noindent
Given a matrix $A$ with $(n+k)$ rows and $n$ columns. 
$A_i$ is the $i$-th row of $A$.
$A_I$ is a submatrix of $A$ having rows with indices in $I$, where $I$ is a subset of $\{1,2,\dots, n+k\}$.
Denote $A_{\{1,2,\dots, n+k\}\backslash\{i\}}$ by $A_{\hat i}$.
Given $i, j\in\{1,2,\dots, n+k\}$ and the subset $I$, where the number of entries of $I$ is $\#I=(n-1)$ and $j\notin I$.
Then the (right) quasi-Pl\"ucker coordinates are given by \cite{gelfand91,gilson07} 
\begin{align*}
r_{ij}^I=-\left|\begin{array}{cc}
A_I&0\\
A_i&\boxed{0}\\
A_j& 1\end{array}
\right|.
\end{align*}
By using quasi-Pl\"ucker coordinates, one could state the following homological relations 
\begin{align}
\left|\begin{array}{ccc}
A&B&C\\
D&f&g\\
E&\boxed{h}&i\end{array}
\right|&=\left|\begin{array}{ccc}
A&B&C\\
D&f&g\\
E&h&\boxed{i}\end{array}
\right|\left|\begin{array}{ccc}
A&B&C\\
D&f&g\\
0&\boxed{0}&1\end{array}
\right|,\label{hm1}\\
\left|\begin{array}{ccc}
A&B&C\\
D&f&\boxed{g}\\
E&h&i\end{array}
\right|&=\left|\begin{array}{ccc}
A&B&0\\
D&f&\boxed{0}\\
E&h&1\end{array}
\right|\left|\begin{array}{ccc}
A&B&C\\
D&f&g\\
E&h&\boxed{i}\end{array}
\right|,\label{hm2}
\end{align}
which were used in this article.
\end{enumerate}

\begin{enumerate}[0]
\item[$\bullet$]  
Derivatives of general quasi-determinants.

\noindent
Let $A$, $B$, $C$ and $d$ be functions of $t$, then
\begin{align}\label{dqd}
\left|\begin{array}{cc}
A &B \\
C &\boxed{d} 
\end{array}
\right|'
=d'-C'A^{-1}B-CA^{-1}B'+CA^{-1}A'A^{-1}B,
\end{align}
where prime denotes the derivative with respect to $t$.
In particular, since we mainly consider   Wronski quasi-determinants in this article, we have the following Wronskian-type derivative formula \cite[eqs. 22, 23]{gilson07}
\begin{align}\label{dqd-2}
\left|\begin{array}{cc}
A &B \\
C &\boxed{d} 
\end{array}
\right|'
=\left|\begin{array}{cc}
A &B \\
C' &\boxed{d'} 
\end{array}
\right|
+\sum_{j=1}^{n}\left|\begin{array}{cc}
A &e_j^\top \\
C &\boxed{0} 
\end{array}
\right|
\left|\begin{array}{cc}
A &B\\
\left(A^j\right)' &\boxed{\left(B^j\right)'} 
\end{array}
\right|,
\end{align}
and
\begin{align}\label{dqd-3}
\left|\begin{array}{cc}
A &B \\
C &\boxed{d} 
\end{array}
\right|'=\left|\begin{array}{cc}
A &B' \\
C &\boxed{d'} 
\end{array}
\right|
+\sum_{j=1}^{n}\left|\begin{array}{cc}
A &\left(A_j\right)' \\
C &\boxed{\left(C_j\right)'} 
\end{array}
\right|
\left|\begin{array}{cc}
A &B\\
e_j &\boxed{0} 
\end{array}
\right|,
\end{align}
where $e_j$ is a block unit vector whose $j$-th position is the unit element, and $A^j$ (respectively $A_j$) is the $j$-th row (respectively column) of $A$.
\end{enumerate}

\end{document}